\documentclass[a4paper,UKenglish]{lipics-v2019}

\usepackage{microtype}

\title{Determinization of Büchi Automata: Unifying the Approaches of Safra and Muller-Schupp}

\titlerunning{Determinization of Büchi Automata}
\author{Christof Löding}{RWTH Aachen University, Ahornstr. 55, 52074 Aachen, Germany}{loeding@cs.rwth-aachen.de}{}{}

\author{Anton Pirogov}{RWTH Aachen University, Ahornstr. 55, 52074 Aachen, Germany}{pirogov@cs.rwth-aachen.de}{https://orcid.org/0000-0002-5077-7497}{This author is supported by the German research council (DFG) Research Training Group 2236 UnRAVeL}

\authorrunning{C. Löding and A. Pirogov}
\Copyright{Christof Löding and Anton Pirogov}



\hideLIPIcs
\nolinenumbers

\newcommand{\triang}{\ensuremath{{\bigtriangleup}}}
\usepackage{scalerel} \def\mcirc{\ensuremath{{\mathbin{\scalerel*{\circ}{k}}}}}
\usepackage{bbding} \newcommand{\hollowstar}{\ensuremath{\text{\FiveStarOpen}}}

\usepackage{thm-restate}

\usepackage{algorithmicx}
\usepackage{algpseudocode}
\usepackage{algorithm}
\algnewcommand{\IIf}[2]{\State #2\ \algorithmicif\ #1}
\algtext*{EndFor}
\algtext*{EndWhile}
\algtext*{EndFunction}

\usepackage{tikz}
\usetikzlibrary{automata,positioning,arrows,backgrounds,calc,decorations.pathreplacing,decorations.pathmorphing,shapes,fit}
\tikzset{every state/.style={minimum size=0pt}}
\newcommand{\convexpath}[2]{
  [
  create hullcoords/.code={
    \global\edef\namelist{#1}
    \foreach [count=\counter] \nodename in \namelist {
      \global\edef\numberofnodes{\counter}
      \coordinate (hullcoord\counter) at (\nodename);
    }
    \coordinate (hullcoord0) at (hullcoord\numberofnodes);
    \pgfmathtruncatemacro\lastnumber{\numberofnodes+1}
    \coordinate (hullcoord\lastnumber) at (hullcoord1);
  },
  create hullcoords
  ]
  ($(hullcoord1)!#2!-90:(hullcoord0)$)
  \foreach [
  evaluate=\currentnode as \previousnode using \currentnode-1,
  evaluate=\currentnode as \nextnode using \currentnode+1
  ] \currentnode in {1,...,\numberofnodes} {
    let \p1 = ($(hullcoord\currentnode) - (hullcoord\previousnode)$),
    \n1 = {atan2(\y1,\x1) + 90},
    \p2 = ($(hullcoord\nextnode) - (hullcoord\currentnode)$),
    \n2 = {atan2(\y2,\x2) + 90},
    \n{delta} = {Mod(\n2-\n1,360) - 360}
    in
    {arc [start angle=\n1, delta angle=\n{delta}, radius=#2]}
    -- ($(hullcoord\nextnode)!#2!-90:(hullcoord\currentnode)$)
  }
}

\newcommand{\op}[1]{\ensuremath{\mathsf{#1}}}
\newcommand{\mc}[1]{\mathcal{#1}}

\newcommand{\parent}{\op{\uparrow}}
\newcommand{\lsibling}{\op{\leftarrow}}
\newcommand{\cut}[1]{#1\text{-}\op{cut}}

\newcommand{\cupdot}{\mathbin{\mathaccent\cdot\cup}}

\newcommand*{\HASAPPENDIX}{} \newcommand{\withappendix}[1]{\ifthenelse{\isundefined{\HASAPPENDIX}}{}{#1}}

\begin{document}
\maketitle

\vspace{7mm}
\begin{abstract}
  Determinization of Büchi automata is a long-known difficult problem, and after the seminal
  result of Safra, who developed the first asymptotically optimal construction from Büchi
  into Rabin automata, much work went into improving, simplifying, or avoiding Safra's
  construction. A different, less known determinization construction was proposed
  by Muller and Schupp. The two types of constructions share some similarities but their
  precise relationship was still unclear.
  In this paper, we shed some light on this relationship by
  proposing a construction from nondeterministic Büchi to
  deterministic parity automata that subsumes both constructions: Our
  construction leaves some freedom in the choice of the successor
  states of the deterministic automaton, and by instantiating these
  choices in different ways, one obtains as particular cases the
  construction of Safra and the construction of Muller and Schupp. The
  basis is a correspondence between structures that are encoded in the
  macrostates of the determinization procedures---Safra trees on one
  hand, and levels of the split-tree, which underlies the Muller and
  Schupp construction, on the other hand. Our construction also allows
  for mixing the mentioned constructions, and opens up new directions
  for the development of heuristics.

\end{abstract}

\vspace{7mm}
\section{Introduction}
\label{sec:introduction}
Büchi automata are finite automata for infinite words, and were initially introduced
to show decidability of the logic S1S \cite{buchi1966symposium}. Infinite words can be
used to model infinite execution traces of reactive, non-terminating systems, and serve as a translation target
from logics like LTL (see, e.g., \cite{SomenziB00,GastinO01}), which is a popular and well-understood specification formalism.
For this reason,  Büchi
automata nowadays play a central role in formal methods like model-checking \cite{baier2008principles} and
runtime-verification \cite{GiannakopoulouH01}, because they can represent all $\omega$-regular languages and are
suitable for efficient algorithmic treatment.
Unfortunately, the simplicity of the Büchi acceptance condition makes
it crucially dependent on nondeterminism, i.e., not every
$\omega$-regular language (or LTL formula) can be accepted by a
deterministic Büchi automaton (see, e.g., \cite{Thomas97}). In some
settings, this nondeterminism causes difficulties, such that
algorithms require a representation of the property by a deterministic
automaton, like in probabilistic model-checking (see, e.g.,
\cite[Section~10.3]{baier2008principles}), or in synthesis (see
\cite{Thomas08} for an overview of the theory, and \cite{MeyerSL18}
for recent developments in practice).

A first determinization procedure that translates nondeterministic
Büchi automata into deterministic automata was presented in
\cite{mcnaughton1966testing}. The first asymptotically optimal and
most well-known determinization construction for Büchi automata is the
construction of Safra \cite{safra1988complexity}. It translates a
nondeterministic Büchi automaton with $n$ states into a deterministic
Rabin automaton with at most $2^{\mathcal{O}(n \log n)}$ states and
  $\mathcal{O}(n)$ sets in the acceptance condition. In applications
  like synthesis, the deterministic automaton is used to build a game
  that inherits as winning condition the acceptance condition of the
  automaton. In the theory of infinite duration games, the parity
  condition plays a central role (see, e.g., the survey
  \cite{VardiW07}). For this reason, Piterman modified Safra's
  construction in order to directly obtain a parity automaton
  \cite{piterman2006nondeterministic}. This construction was
  reformulated in \cite{schewe2009tighter}, where also a tighter
  analysis of its state complexity is given with an upper bound of $\mathcal{O}(n!^2)$ for the number of states. A similar
  construction is presented in \cite{redziejowski2012improved},
  adapted to the translation of $\omega$-regular expressions directly
  into parity automata.

 It is known that the Safra construction is essentially optimal
 \cite{ColcombetZ09}, so there is no hope of significantly improving the worst-case
 upper bounds of the known constructions.  However, the data structure
 of Safra trees (or history trees) that is used for the states of the
 deterministic automata, is challenging to deal with in
 implementations. Therefore, alternative approaches for
 determinization have been studied, leading to a family of
 constructions that are based on a construction by Muller and Schupp,
 which appeared in \cite{muller1995simulating} as a by-product of a
 translation from alternating to non-deterministic tree automata.  An
 explicit description of the construction specifically for
 determinization of Büchi automata is presented in
 \cite{kahler2008complementation}.
   A refinement of that construction is presented in
 \cite{fogarty2015profile}, in which the states of the deterministic
 automaton are no longer represented as trees but as ordered and
 labelled tuples of sets.

 The two approaches of Safra and Muller-Schupp show some similarities,
 as pointed out in the conclusion of \cite{fogarty2015profile}, but
 from the existing formulations of the constructions, their precise
 relationship is not clear.

 In this paper, we provide a construction for transforming
 nondeterministic Büchi automata into deterministic parity automata
 that cleanly explains the connections between the approaches of Safra
 and Muller-Schupp. It turns out that both types of constructions can
 be formulated on the same data structure, which can either be
 understood as ordered tuples of sets in which each set has an
 additional rank (a natural number), or as Safra trees in which each
 node has an additional rank (the same structure is essentially used in
 the constructions from \cite{piterman2006nondeterministic} and
 \cite{schewe2009tighter}).
 The transitions are defined in terms of a sequence of simple
 operations, and it turns out that the two constructions only differ
 in one of these operations.
In summary, our contributions are the following:
 \begin{itemize}
 \item We provide a new and relatively simple formulation of a
   Muller-Schupp style determinization construction that yields
   deterministic parity automata. Compared to previous constructions
   from \cite{kahler2008complementation} and \cite{fogarty2015profile},
   we encode less information in the states, and obtain a construction
   that has the same worst-case upper bound as the Safra style
   constructions.
 \item We extend our Muller-Schupp style construction by introducing
   a degree of freedom in the choice of the successor states. This freedom
   can be used to make the construction correspond to     Safra's construction as presented in
   \cite{piterman2006nondeterministic} and
   \cite{schewe2009tighter}. We therefore obtain a construction
   that unifies the approaches of Safra and Muller-Schupp in one
   general construction.  Furthermore, the freedom in the choice of
   the successors of transitions also yields new ways of obtaining
   deterministic parity automata, and can be used in implementations
   as a heuristic to reduce the state space of the resulting
   automaton.
 \end{itemize}

This work is organized as follows. After introducing the basic notations in
Section~\ref{sec:preliminaries}, we present the new variant of the Muller-Schupp
construction in Section~\ref{sec:mullerschupp}, and then briefly review Safra's
construction in Section~\ref{sec:safra}. We explain the structural relationship
between those two constructions in Section~\ref{sec:bijection}, and finally introduce our
generalized construction as a simple extension of the presented Muller-Schupp
construction in Section~\ref{sec:construction}. In Section~\ref{sec:discussion} we discuss and conclude.

 \section{Preliminaries}
\label{sec:preliminaries}

First we briefly review basic definitions concerning $\omega$-automata and $\omega$-languages.
If $\Sigma$ is a finite alphabet, then $\Sigma^\omega$ is the set of all infinite
\emph{words} $w=w_0w_1\ldots$ with $w_i\in\Sigma$. For $w\in \Sigma^\omega$ we denote by
$w(i)$ the $i$-th symbol $w_i$. For convenience, we write $[n]$ for the set of natural
numbers $\{1,\ldots,n\}$.
A \emph{Büchi automaton} $\mc{A}$ is a tuple $(Q, \Sigma, \Delta, Q_0,
F)$, where $Q$ is a finite set of states, $\Sigma$ a finite alphabet, $\Delta \subseteq
Q\times \Sigma \times Q$ is the transition relation and $Q_0, F \subseteq Q$ are the sets
of initial and accepting states, respectively. When $Q$ is understood and $X\subseteq Q$, then
$\overline{X} := Q \setminus X$.
We write $\Delta(p,x) := \{q \mid (p,x,q) \in \Delta \}$ to denote the set of
\emph{successors} of $p$ on symbol $x$ and $\Delta(P,x)$ for $\bigcup_{p\in
P}\Delta(p,x)$.
A \emph{run} of an automaton on a word $w\in \Sigma^\omega$ is an infinite sequence of
states $q_0, q_1, \ldots$ starting in some $q_0 \in Q_0$ such that $(q_i, w(i), q_{i+1})
\in \Delta$ for all $i\geq 0$.
An automaton is \emph{deterministic} if $|Q_0|=1$ and $|\Delta(p,x)|\leq 1$ for all $p\in Q, x\in\Sigma$, and
\emph{non-deterministic} otherwise. In this work, we assume Büchi automata to be
non-deterministic and refer to them as NBA.
A \emph{transition-based deterministic parity automaton} (TDPA) is a deterministic
automaton $(Q,\Sigma,\Delta,Q_0,c)$ where instead of $F\subseteq Q$ there is a
\emph{priority function} $c : \Delta \to \mathbb{N}$ assigning a natural number to each
transition.

A run of an NBA is \emph{accepting} if it contains infinitely many accepting states.
A run of a TDPA is accepting if the smallest priority that appears infinitely often on transitions along the
run is even.
An automaton $\mc{A}$ \emph{accepts} $w\in \Sigma^\omega$ if there exists an accepting run on $w$, and
the language $L(\mc{A}) \subseteq \Sigma^\omega$ \emph{recognized} by $\mc{A}$ is the set
of all accepted words. To avoid confusion, we sometimes refer to states of TDPA that we
construct as \emph{macrostates} to distinguish them from the states of the underlying
Büchi automaton.

 \section{A Simplified Muller-Schupp Construction}
\label{sec:mullerschupp}

The essential idea for determinization using the Muller-Schupp
construction is the following: given some Büchi automaton $\mc{A}$ and
input word $w$, the resulting deterministic automaton conceptually
traverses a specific run-tree of $\mc{A}$ on $w$, called reduced
split-tree in \cite{kahler2008complementation}, and tracks enough
information to decide whether an infinite path with a specific shape
exists in this tree. Such a path is known to exist if and only if $w$
is accepted by $\mc{A}$. The construction presented in
\cite{kahler2008complementation} uses a structure called contraction
trees in order to track the relevant information. This has been
simplified in \cite{fogarty2015profile} to macrostates that consist of
an ordered tuple of disjoint sets of Büchi states, and two preorders
over the states appearing in the tuple.

In this section, we further simplify the structure of the macrostates for the
deterministic automaton to ordered tuples of disjoint sets of Büchi
states, and a single additional linear order on these sets (formally
expressed as a ranking function that assigns to each set a natural
number). This also results in a relatively simple transition function
on the macrostates.

The \emph{reduced split-tree} $t^{rs}(\mc{A},w)$ for NBA $\mc{A}$ and
word $w\in\Sigma^\omega$ is an ordered infinite tree in which the
nodes are labelled by state-sets, and each node has at most two
successors. Formally, it is constructed as follows. The first level of
the tree consists of the root node labelled by the initial states
$Q_0$. To construct level $i+1$ from level $i$, for each node at level
$i$ labelled by set $S$ of states, let the \emph{left child} of $S$ be
labelled by $\Delta(S, w(i))\cap F$ and the \emph{right child} by
$\Delta(S, w(i)) \cap \overline{F}$, i.e., accepting and non-accepting
successor states are separated.  Then keep only the leftmost (wrt. the
natural ordering of neighbors) occurrence of each state in the level
and finally remove nodes labelled by $\emptyset$.  Clearly, because of
the normalization, the number of nodes on each level can be at most
$|Q|$. An example of a reduced split-tree is shown in
Figure~\ref{fig:splittree}. We call an infinite path in the tree that
takes the left branch infinitely often a \emph{left-path}. Reduced
split-trees have the following useful property:

\begin{lemma}[\cite{kahler2008complementation}, Lemma 2]
  \label{lem:normisgood}
  $\mc{A}$ accepts $w \Longleftrightarrow t^{rs}(\mc{A}, w)$ has a left-path.
\end{lemma}

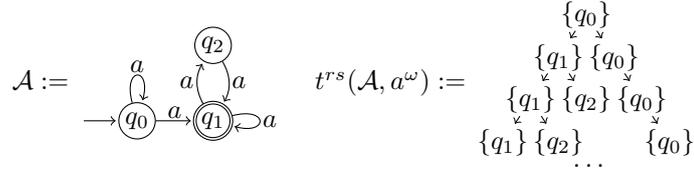
\begin{figure}[t]
  \begin{center}
        $\mc{A} :=$
    \begin{tikzpicture}[baseline={([yshift=-.5ex]current bounding box.center)},
      shorten >=1pt,node distance=1cm,inner sep=1pt,on grid,auto]
      \node[state,initial,initial text=]   (q0)               {$q_0$};
      \node[state,accepting] (q1) [right=of q0] {$q_1$};
      \node[state]           (q2) [above=of q1] {$q_2$};
        \path[->]
        (q0) edge [loop above] node {$a$} (q0)
        (q1) edge [loop right] node {$a$} (q1)
        (q0) edge node {$a$} (q1)
        (q1) edge [bend left] node {$a$} (q2)
        (q2) edge [bend left] node {$a$} (q1)
        ;
    \end{tikzpicture}
    \quad
    $t^{rs}(\mc{A}, a^\omega) :=$
    \begin{tikzpicture}[baseline={([yshift=-.5ex]current bounding box.center)},
      shorten >=1pt,node distance=8mm,inner sep=1pt,on grid,auto]
      \node[]   (l0s1)                 {$\{q_0\}$};
      \node[]   (l1s1) [below left=of l0s1, xshift=2mm] {$\{q_1\}$};
      \node[]   (l1s2) [below right=of l0s1, xshift=-2mm] {$\{q_0\}$};
      \node[]   (l2s1) [below left=of l1s1, xshift=2mm] {$\{q_1\}$};
      \node[]   (l2s2) [below right=of l1s1, xshift=-2mm] {$\{q_2\}$};
      \node[]   (l2s3) [below right=of l1s2, xshift=-2mm] {$\{q_0\}$};
      \node[]   (l3s1) [below left=of l2s1, xshift=2mm] {$\{q_1\}$};
      \node[]   (l3s2) [below right=of l2s1, xshift=-2mm] {$\{q_2\}$};
      \node[]   (l3s3) [below right=of l2s3, xshift=-2mm] {$\{q_0\}$};
      \node[]   (l4) [below=of l3s2,yshift=5mm,xshift=4mm] {$\ldots$};
        \path[->]
        (l0s1) edge [] node {} (l1s1)
        (l0s1) edge [] node {} (l1s2)
        (l1s1) edge [] node {} (l2s1)
        (l1s1) edge [] node {} (l2s2)
        (l1s2) edge [] node {} (l2s3)
        (l2s1) edge [] node {} (l3s1)
        (l2s1) edge [] node {} (l3s2)
        (l2s3) edge [] node {} (l3s3)
        ;
    \end{tikzpicture}
  \end{center}

  \caption{
            Example of a reduced split-tree $t^{rs}(\mc{A}, a^\omega)$ of an NBA $\mc{A}$.
            It has an infinite path representing the run $q_0^\omega$ and a
            left-path representing the run $q_0 q_1^\omega$, from which finite paths
            $q_0 q_1^* q_2$ branch off.
          }
  \label{fig:splittree}

\end{figure}

In the following, we identify nodes in the same level with their
label sets. To obtain a deterministic automaton,  we augment the nodes of the reduced
split-tree with number tokens that we call \emph{(age-)ranks}, which are used to infer a
left-path.

The new macrostates in the deterministic automaton represent levels of
reduced split-trees and consist of a tuple of disjoint non-empty sets $t := (S_1, S_2, \ldots, S_n)$
equipped with a bijection $\alpha : [n] \to [n]$ satisfying $\alpha(n) = 1$, which assigns to each set $S_i$
the rank $\alpha(i)$.
We call a pair $(\alpha,t)$ that satisfies these constraints \emph{ranked slice} and
we call it \emph{pre-slice}, if $t$ contains empty sets or $\alpha$ is not a bijection.
Notice that all macrostates are ranked slices, whereas pre-slices occur only during intermediate steps.
We introduce the following useful notations to work with ranked slices.
Let $|t|:=n$ and $Q_t := \bigcup_{i=1}^{|t|} S_{i}$.
The function $\op{idx} : Q_{t} \to [|t|]$ maps each state $q\in Q_{t}$ to the tuple index
$i$ such that $q \in S_{i}$, and by $\alpha(q)$ we denote $\alpha(\op{idx}(q))$ for $q \in
Q_t$.

When reading symbol $x \in \Sigma$ in macrostate $(\alpha, t)$, the
successor macrostate $(\alpha', t')$ is obtained by a sequence of
successive operations $\op{step}$, $\op{prune}$ and $\op{normalize}$,
where, roughly,

\begin{itemize}
\item $\op{step}$ interprets $t$ as nodes on a reduced
  split-tree level and calculates the next level sets,
\item $\op{prune}$ removes the empty sets produced by $\op{step}$,
  reassigning ranks in a specific way, and
\item $\op{normalize}$ just turns the ranking function obtained after
  $\op{prune}$ into a bijection again.
\end{itemize}

Below, we formally define these operations ($\op{step}$ and
$\op{prune}$ are illustrated in Figure~\ref{fig:mullerschupp}).

\begin{figure}
  \begin{center}
    \begin{tikzpicture}[baseline={([yshift=-.5ex]current bounding box.center)},
      shorten >=1pt,node distance=1cm,inner sep=1pt,on grid,auto]
      \node[]   (t)                 {$(\alpha, t)$};
      \node[]   (eq)  [right of=t]  {$= ($};

      \node[]   (s1) [right of=eq, xshift=4mm]    {${S_1}$};
      \node[]   (a1) [above right= 1.1mm and 4mm of s1] {$\scriptstyle \alpha(1)$};
      \node[]   (d1) [below of=s1, yshift=3mm]    {$\scriptstyle \Delta_t(S_1,x)$};
      \node[]   (s2) [right of=s1, xshift=15mm]   {${S_2}$};
      \node[]   (a2) [above right= 1.1mm and 4mm of s2] {$\scriptstyle \alpha(2)$};
      \node[]   (d2) [below of=s2, yshift=3mm]    {$\scriptstyle \Delta_t(S_2,x)$};
      \node[]   (s3) [right of=s2, xshift=17mm]   {${S_3}$};
      \node[]   (a3) [above right= 1.1mm and 4mm of s3] {$\scriptstyle \alpha(3)$};
      \node[]   (d3) [below of=s3, yshift=3mm]    {$\scriptstyle \Delta_t(S_3,x)$};
      \node[]   (s4) [right of=s3, xshift=9mm]   {$\ldots$};
      \node[]   (sn) [right of=s4, xshift=6mm]   {${S_n}$};
      \node[]   (an) [above right= 1.1mm and 4mm of sn] {$\scriptstyle \alpha(n)$};
      \node[]   (dn) [below of=sn, yshift=3mm]    {$\scriptstyle \Delta_t(S_n,x)$};

      \node[]   (cl)  [right of=sn, xshift=10mm]  {$)$};

      \node[]   (t2) [below of=t, yshift=-08mm]                {$(\hat{\alpha}, \hat{t})$};
      \node[]   (eq2)  [right of=t2]  {$= ($};

      \node[]   (s12) [right of=eq2,xshift=-3mm]               {${\hat{S}_1}$};
      \node[]   (a12) [above right= 1.1mm and 4mm of s12] {$\color{black!60}\scriptstyle n+1$};
      \node[]   (s22) [right of=s12, xshift=3mm]   {${\hat{S}_2}$};
      \node[]   (a22) [above right= 1.1mm and 4mm of s22] {$\scriptstyle \alpha(1)$};
      \node[]   (s32) [right of=s22, xshift=3mm]   {${\hat{S}_3}$};
      \node[]   (a32) [above right= 1.1mm and 4mm of s32] {$\color{black!60}\scriptstyle n+1$};
      \node[]   (s42) [right of=s32, xshift=3mm]   {${\hat{S}_4}$};
      \node[]   (a42) [above right= 1.1mm and 4mm of s42] {$\scriptstyle \alpha(2)$};
      \node[]   (s52) [right of=s42, xshift=3mm]   {${\hat{S}_5}$};
      \node[]   (a52) [above right= 1.1mm and 4mm of s52] {$\color{black!60}\scriptstyle n+1$};
      \node[]   (s62) [right of=s52, xshift=3mm]   {${\hat{S}_6}$};
      \node[]   (a62) [above right= 1.1mm and 4mm of s62] {$\scriptstyle \alpha(3)$};
      \node[]   (s72) [right of=s62, xshift=3mm]   {$\ldots$};
      \node[]   (s82) [right of=s72, xshift=3mm]   {${\hat{S}_{2n-1}}$};
      \node[]   (a82) [above right= 1.1mm and 4mm of s82] {$\color{black!60}\scriptstyle n+1$};
      \node[]   (sn2) [right of=s82, xshift=3mm]   {${\hat{S}_{2n}}$};
      \node[]   (an2) [above right= 1.1mm and 4mm of sn2] {$\scriptstyle \alpha(n)$};

      \node[]   (cl2)  [below of=cl, yshift=-08mm]  {$)$};

      \node[]   (t3) [below of=t2, yshift=-5mm]                {$(\tilde{\alpha}, \tilde{t})$};
      \node[]   (eq3)  [right of=t3]  {$= ($};

      \node[]   (s13) [below of=s22,yshift=-5mm]               {${\tilde{S}_1}$};
      \node[]   (a13) [above right= 1.1mm and 5mm of s13] {$\scriptstyle \tilde{\alpha}(1)$};
            \node[]   (m13) [above= 5mm of a13,xshift=8mm] {$\scriptstyle \min$};
      \node[]   (s23) [below of=s52, yshift=-5mm]   {${\tilde{S}_2}$};
      \node[]   (a23) [above right= 1.1mm and 5mm of s23] {$\scriptstyle \tilde{\alpha}(2)$};
            \node[]   (m23) [above= 5mm of a23,xshift=8mm] {$\scriptstyle \min$};
      \node[]   (sn3) [below of=s82, yshift=-5mm]   {${\tilde{S}_{\tilde{n}}}$};
      \node[]   (an3) [above right= 1.1mm and 5mm of sn3] {$\scriptstyle \tilde{\alpha}(\tilde{n})$};
            \node[]   (mn3) [above= 5mm of an3,xshift=5mm] {$\scriptstyle \min$};

      \node[]   (cl3)  [below of=cl2, yshift=-5mm]  {$)$};

      \path[->] (t) edge node [xshift=1mm] {$\op{step}$} (t2);
      \path[->] (t2) edge node [xshift=1mm] {$\op{prune}$} (t3);

      \path[->] (m13) edge (a13)
                (m23) edge (a23)
                (mn3) edge (an3)
                (a22) edge [bend right] (m13)
                (a32) edge [bend left] (m13)
                (a42) edge [bend left] (m13)
                (a52) edge [bend right] (m23)
                (a62) edge [bend left] (m23)
                (s72) edge [bend left] (m23)
                (a82) edge [bend right,out=0] (mn3)
                (an2) edge [bend left] (mn3)
                ;
      \path[dotted,->]
                (s22) edge [swap] node {$\scriptstyle\neq\emptyset$} (s13)
                (s52) edge [swap] node {$\scriptstyle\neq\emptyset$} (s23)
                (s82) edge [swap] node {$\scriptstyle\neq\emptyset$} (sn3)
                ;
      \path[->]
                (a1) edge [bend left] (a22)
                (a2) edge [bend left] (a42)
                (a3) edge [bend left] (a62)
                (an) edge [bend left] (an2)
                                                                ;
      \path[dotted,->]
                (s1) edge (d1)
                (s2) edge (d2)
                (s3) edge (d3)
                (sn) edge (dn)
                (d1) edge [swap] node {$\scriptstyle \cap F$} (s12)
                (d1) edge node {$\scriptstyle \cap \overline{F}$} (s22)
                (d2) edge [swap] node {$\scriptstyle \cap F$} (s32)
                (d2) edge node {$\scriptstyle \cap \overline{F}$} (s42)
                (d3) edge [swap] node {$\scriptstyle \cap F$} (s52)
                (d3) edge node {$\scriptstyle \cap \overline{F}$} (s62)
                (dn) edge [swap] node {$\scriptstyle \cap F$} (s82)
                (dn) edge node {$\scriptstyle \cap \overline{F}$} (sn2)
                ;

      \node[]   (c11)  [left=6mm of a22]  {};
      \node[]   (c21)  [above=0mm of a42]  {};
      \node[]   (c31)  [below=8mm of c21]  {};
      \node[]   (c41)  [left=2mm of s13]  {};

      \node[]   (c12)  [left=6mm of a52]  {};
      \node[]   (c22)  [above=1mm of s72]  {};
      \node[]   (c32)  [below=8mm of c22]  {};
      \node[]   (c42)  [left=2mm of s23]  {};

      \node[]   (c13)  [left=6mm of a82]  {};
      \node[]   (c23)  [above=0mm of an2]  {};
      \node[]   (c33)  [below=8mm of c23]  {};
      \node[]   (c43)  [left=2mm of sn3]  {};
      \begin{scope}[on background layer]
        \draw[fill=black!5,draw=white] \convexpath{c11,c21,c31,s13,c41}{3mm};
        \draw[fill=black!5,draw=white] \convexpath{c12,c22,c32,s23,c42}{3mm};
        \draw[fill=black!5,draw=white] \convexpath{c13,c23,c33,sn3,c43}{3mm};
      \end{scope}
    \end{tikzpicture}
  \end{center}
  \caption{
            Abstract illustration of $\op{step}$ and $\op{prune}$ in a Muller-Schupp
            transition on some $x\in \Sigma$.
            The superscripts represent the assigned ranks. First, $\op{step}$ calculates the
            normalized successors, separating accepting from non-accepting states and
            passing the parent rank on to the right child. In the illustration, we assume
            that the sets $\hat{S}_i \neq \emptyset$ for $i\in\{2,5,2n-1\}$, i.e., $x_1 =
            2, x_2 = 5, x_{\tilde{n}} = x_3 = 2n-1$. Then $\op{prune}$ keeps sets at these positions for the
            resulting tuple $\tilde{t}$, and $\tilde{\alpha}$ is obtained by taking the minimum of the
            ranks given by $\hat\alpha$ in the ranges spanning from one $x_i$ up to the position before
            $x_{i+1}$. Finally $t':=\tilde{t}$ and  $\tilde\alpha$ is normalized to
            $\alpha'$, while preserving strict ordering between positions wrt.\
            $\tilde\alpha$. The dotted edges connect parent sets (in the top row) and
            resulting left/right children sets (bottom row) in the conceptual reduced
            split-tree, the solid edges show the movements of the rank values assigned to
            the sets.
            }

  \label{fig:mullerschupp}
\end{figure}
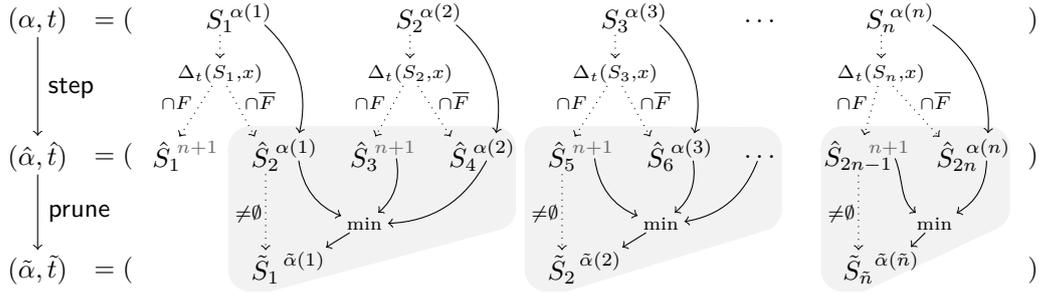

First we describe $\op{step}$, which constructs the next level of the reduced split-tree
and passes each existing rank on to the respective right child. Let \[\Delta_t(q,x) :=
\Delta(q,x)\setminus \Delta(\bigcup_{i=1}^{\op{idx}(q)-1} S_i,x),\] restricting for each
state $q\in Q_t$ the successors to only those which are not reached by some other state
located in a set to the left of $q$. Then, for each node $S_i$ let $\hat{S}_{2i-1} :=
\Delta_t(S_i, x)\cap F$ be the \emph{left child} and $\hat{S}_{2i} := \Delta_t(S_i,
x)\cap \overline{F}$ the \emph{right child}, containing the accepting and non-accepting
normalized successors, respectively. Let $\hat{\alpha}(2i) := \alpha(i)$ and
$\hat{\alpha}(2i-1) := n+1$, i.e., the right children inherit the rank of the parent and the
left children all get the same new maximal rank $n+1$, resulting in a pre-slice $(\hat{\alpha},\hat{t})$.

Intuitively, in the $\op{prune}$ operation, all ranks that mark empty sets after
$\op{step}$ are relocated onto the closest non-empty set to the left (or removed, if no
such set exists). When multiple ranks occupy the same set, then the smallest one is
preserved. Ranks that moved to the left in this way and are not removed, indicate a good
(green) event, whereas ranks which
were removed indicate a bad (red) event.

Formally, let $x_1 < x_2 < \ldots < x_{\tilde{n}}$ be the increasing sequence of all indices such that
$\hat{S}_{x_j} \neq \emptyset$.
Then $\op{prune}$ returns $(\tilde{\alpha}, \tilde{t})$ with the tuple $\tilde{t} := (\hat{S}_{x_1},
\ldots, \hat{S}_{x_{\tilde{n}}})$ without empty sets, where $\tilde{\alpha}$ is defined as
$\tilde{\alpha}(i) := \min\{ \hat{\alpha}(j) \mid x_i \leq j < x_{i+1} \}$
with $x_{\tilde{n}+1} := |\hat{t}|+1$.

The set of \emph{green} ranks is given by $G :=
\op{img}(\tilde\alpha)\cap\{\hat\alpha(j)\mid \hat{S}_j = \emptyset
\}$, where $\op{img}(\tilde\alpha)$ denotes the image of
$\tilde{\alpha}$. These are the ranks that
mark empty sets after $\op{step}$ and are not removed by $\op{prune}$.
The set of \emph{red} ranks given by $R := \op{img}(\hat\alpha) \setminus
\op{img}(\tilde\alpha)$ contains the ranks that were not preserved during $\op{prune}$.
The set of \emph{active ranks} is $A := G \cupdot R$.
Let $k := \min A$ (or $k := |Q|+1$ if $A=\emptyset$)
denote the \emph{dominating} rank of the transition, i.e., the smallest active
rank. We define the priority $p$ of the transition as $2k$ if $k\in G$ and $2k - 1$ otherwise.

The function $\tilde{\alpha}$ might assign the same rank to several
sets, and it might have gaps (unused rank values between
used ones). So finally, $\op{normalize}$ returns $(\alpha',t')$ with $t' := \tilde{t}$ and a final bijective ranking function
$\alpha' : [|t'|]\to[|t'|]$ such that $\tilde{\alpha}(j) < \tilde{\alpha}(k) \Rightarrow \alpha'(j) <
\alpha'(k)$ for all $j,k \in \{1,\ldots,|t'|\}$, i.e., a total order which is compatible
with the preorder induced by $\tilde{\alpha}$. If there are several
such ranking functions $\alpha'$, then any of these works.

A TDPA $\mc{B}$ is obtained by taking the initial state
$(\alpha_0, t_0)$ with $t_0 := (Q_0), \alpha_0(1) := 1$ and a transition function that
picks for each state a valid successor that satisfies the description above, and assigns
the corresponding priority $p$ to the edge.  Observe that by
construction, the sequence of states visited along some word $w\in \Sigma^\omega$ from the
initial state represents exactly the levels of $t^{rs}(\mc{A},w)$,
marked with ranks.

\begin{theorem} \label{thm:muller-schupp}
For a given NBA with $n$ states, the TDPA obtained by the Muller-Schupp construction accepts the same language as the NBA, and its number of states  is in $\mathcal{O}(n!^2)$.
\end{theorem}

\vspace{5mm}
The correctness follows from the correctness of the generalized construction presented in
Section~\ref{sec:construction}.
The claim on the state complexity directly follows from
the upper bound given in \cite[Proposition~2]{schewe2009tighter}, and the bijection
between the set of ranked slices and the set of ranked Safra trees presented in
Section~\ref{sec:bijection}.
 \section{Sketch of the Safra Construction}
\label{sec:safra}

In this section, we  roughly illustrate the used structures and operations of the
Safra construction along the lines of \cite{piterman2006nondeterministic,schewe2009tighter}, so that we can demonstrate
its relationship with the Muller-Schupp construction in the next section. As before, $\mc{A}$ is an NBA with the usual components.

\begin{figure}[t]
  \begin{center}
    $\mc{A} :=$
    \begin{tikzpicture}[baseline={([yshift=-.5ex]current bounding box.center)},
      shorten >=1pt,node distance=1cm,inner sep=1pt,on grid,auto]
      \node[state,initial,initial text=]   (q0)               {$q_0$};
      \node[state,accepting]   (q2) [below=of q0] {$q_2$};
      \node[state] (q1) [right=of q0] {$q_1$};
      \node[state,accepting] (q3) [right=of q2] {$q_3$};
        \path[->]
        (q0) edge [loop right,looseness=5] node {$a$} (q0)
        (q2) edge [loop right,looseness=5] node {$a$} (q2)
        (q0) edge [] node {$a$} (q2)
        (q0) edge [] node {$a$} (q3)
        (q3) edge [swap] node {$a$} (q1)
        ;
    \end{tikzpicture}
    \quad\quad\quad
    \begin{tikzpicture}[baseline={([yshift=-.5ex]current bounding box.center)},
      shorten >=1pt,node distance=8mm,inner sep=1pt,on grid,auto]
      \node[]   (q0)               {$\{q_0{\color{gray},q_1,q_2}\}^1$};
      \node[]   (q1) [below left=of q0]              {$\{q_1\}^2$};
      \node[]   (q2) [below right=of q0]             {$\{q_2\}^3$};
        \path[->]
        (q0) edge [] node {} (q1)
        (q0) edge [] node {} (q2)
        ;
    \end{tikzpicture}
    $\overset{a}{\to}$
    \begin{tikzpicture}[baseline={([yshift=-.5ex]current bounding box.center)},
      shorten >=1pt,node distance=1cm,inner sep=1pt,on grid,auto]
      \node[]   (q0)               {$\{q_0\}^1$};
      \node[]   (q1) [below left=of q0]              {$\emptyset^2$};
      \node[]   (q2) [below=of q0,yshift=3mm]        {$\emptyset^3$};
      \node[]   (q3) [below right=of q0,xshift=5mm]             {$\{q_3\}$};
      \node[]   (q4) [below=of q1,yshift=4mm]              {$\emptyset$};
      \node[]   (q5) [below=of q2,yshift=4mm]              {$\{q_2\}^{\ }$};
        \path[->]
        (q0) edge [] node {} (q1)
        (q0) edge [] node {} (q2)
        (q0) edge [] node {} (q3)
        (q1) edge [] node {} (q4)
        (q2) edge [] node {} (q5)
        ;
    \end{tikzpicture}
    $\to$
    \begin{tikzpicture}[baseline={([yshift=-.5ex]current bounding box.center)},
      shorten >=1pt,node distance=8mm,inner sep=1pt,on grid,auto]
      \node[]   (q0)               {$\{q_0{\color{gray},q_2,q_3}\}^1$};
      \node[]   (q1) [below left=of q0]              {$\{q_2\}^2$};
      \node[]   (q2) [below right=of q0]             {$\{q_3\}^3$};
        \path[->]
        (q0) edge [] node {} (q1)
        (q0) edge [] node {} (q2)
        ;
    \end{tikzpicture}
  \end{center}

  \caption{ Example of a Safra-tree transition on letter $a$, based on NBA $\mc{A}$.
  The LIR position of nodes is depicted as superscript of the sets. The ``redundant''
  states that are implicit in our definition are depicted in gray in the initial and
  resulting tree. In the intermediate step, the tree is depicted after calculating and
  pruning successor state sets. In the final tree the remaining actions are performed and
  LIR positions are updated. The transition has a red event for LIR position 2 and a green event
  for position 3. Because of the removal of the node at position 2 in the LIR, the node that
  originally was at position $3$ moved up, whereas the fresh node labelled by $\{q_3\}$ comes last.
  }
  \label{fig:safra}
\end{figure}
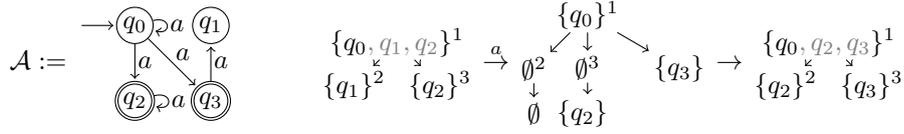

A \emph{Safra tree} is a finite ordered tree with non-empty state-sets as labels. Usually,
it is required that a parent is labelled by a strict superset of all states in its subtree and
siblings are labelled by pairwise disjoint sets. We  use the equivalent
requirement that \emph{all} labels in the tree are pairwise disjoint, i.e., refrain from
listing states in the parent label which are already present in some descendant. One can easily
reconstruct the ``full'' label set of a node wrt.\ the classical definition by taking the
union of all the labels in its subtree.
To obtain parity automata, each node of the Safra tree is associated
with a number from $\{1,\ldots, n\}$, where $n$ is the number of nodes
in the Safra tree \cite{piterman2006nondeterministic}. These numbers satisfy the property that
parent nodes have smaller numbers than their children, and a node has a smaller number than its right sibling.
The numbers
correspond to the ranks that we use in Section~\ref{sec:mullerschupp},
and we therefore refer to Safra trees in combination with these
numbers as \emph{ranked Safra trees}. Two ranked Safra trees are shown
in Figure~\ref{fig:safra} (and an intermediate tree in the middle).

In \cite{schewe2009tighter}, a slightly different representation is
used based on a \emph{later introduction record} (LIR), which just
lists the tree nodes in their introduction order, i.e., nodes appear
in this list after parents and older siblings (in this representation,
nodes have canonical names depending only on their position in the
tree). Safra trees with LIR directly correspond to ranked Safra trees
by annotating each tree node with its position in the LIR.

A transition on symbol $x\in \Sigma$ is constructed as follows (see Figure~\ref{fig:safra} for an example).
First, for each label set $S$, the set $S' := \Delta(S, x)$ of successor states is
calculated. After this, each node gets a fresh right-most child, and the accepting states in $S'$, that is $S' \cap
F$, are moved into the label of this child.
Then, disjointness is ensured by keeping of each state only the copy which is located at the
deepest node along the leftmost branch where that state occurs (this stage is represented by the middle tree in Figure~\ref{fig:safra}).
If now some internal node has an empty label, but a non-empty subtree (a good event for
the node), its subtree is collapsed into a single node by removing all descendants and
moving the states in their labels into the parent label.
Finally, all remaining sets that are labelled by $\emptyset$ are removed (being removed is
a bad event for a node). In the following, we refer to good and bad events
as \emph{green} and \emph{red}, respectively. The priority for the transition is derived from
the green and red
events, which are associated with the relative position of the corresponding nodes in the
LIR. The LIR for the new tree is obtained by deleting removed nodes from the LIR and
appending fresh nodes that remain in the resulting tree in arbitrary order.

 \section{From Safra-trees to ranked slices and back}
\label{sec:bijection}

In this section we state the key observation that was the starting
point of this work: there is a bijection between the set of ranked
slices and the set of ranked Safra trees. From a ranked Safra tree,
one obtains the ranked slice by simply listing the nodes of the Safra
tree by a depth-first post-order traversal (i.e., a parent processed after
all its children).  We formalize this relationship below, and then explain
that the transitions defined in the Muller-Schupp construction and in
the Safra construction are very similar, which then leads to the unified construction.

Let $(\alpha, t)$ be a ranked slice with $t=(S_1, \ldots, S_n)$.
The tuple index of the \emph{parent} of $S_{i}$ is the closest index to the right of $i$ that has
a smaller rank and is formally defined as
$\parent(i) := \min_{i < k \leq n}\{k \mid \alpha(k) < \alpha(i)\}$. As we require by definition of ranked slices that the right-most position in the tuple always has rank 1, this is the only position in the tuple for which the parent is undefined.
The ordered tree induced by $\parent$, with siblings in tuple index order,
is called the \emph{rank tree} of $(\alpha,t)$.
The tuple index of the \emph{left subtree boundary} of $S_{i}$ is the closest index to the left
with a smaller rank, and is denoted by
$\lsibling(i) := \max_{1\leq k < i}\{k \mid \alpha(k) < \alpha(i)\}$
or $0$ if no such index exists. It points to either the direct left sibling of $i$, or
the left sibling of the closest ancestor, if one exists. Effectively, $\lsibling(i)$ is
the closest neighbor to the left which is not a descendant of $i$. As children by definition are
always to the left of their parents, every node at indices
$\lsibling(i)+1,\ldots,i$ is in the subtree of $i$.

For an example, consider the tuple $(\{q_3\}^4, \{q_1\}^2, \{q_2\}^3, \{q_0\}^1)$, where the superscripts denote
the assigned rank (e.g., $\alpha(1) = 4$). The rightmost position $4$ of the tuple is the root of the tree.
For the positions 2 and 3, which have rank 2 and 3 respectively, the next position to the right with a smaller rank is in both cases position 4, i.e.,
$\parent(2) = \parent(3) = 4$. Finally, position 1 in the tuple has position 2 as parent, i.e., $\parent(1)=2$.
The discussed tuple is depicted with the parent edges at the bottom right of Figure~\ref{fig:saframullerschupp}.
There is also one non-trivial left subtree boundary in this tuple, assigned by
$\lsibling(3) = 2$, i.e. index 2 is not in the subtree of index 3, and in this case is an
actual left sibling of index 3.

We use the notation $\parent_\alpha := \alpha\circ\parent\circ\alpha^{-1}$ to denote the
parent rank of another rank directly, without mentioning the indices in the tuple. In the previous example, we have  $\parent_\alpha(4) = 2$, and $\parent_\alpha(2) = \parent_\alpha(3) = 1$.
We identify the age-ranks $\alpha(i)$ as nodes of the tree, while each set $S_{i}$
determines the label of the node $\alpha(i)$, called \emph{hosted set}.
We write $S_i^\downarrow := \bigcup_{k=\lsibling(i)+1}^{i} S_k$ for the \emph{subtree
set} of node $\alpha(i)$.
\begin{definition}
  \label{def:bijection}
  Let \textsf{safra2slice} be the mapping which takes a ranked Safra tree and returns
  $(\alpha, t)$, with $t:=(S_1, \ldots, S_n)$ being the label sets of the nodes
  in depth-first post-order (i.e., a parent processed after all its children) traversal order
  and ranking $\alpha$ defined by the ranks of the corresponding Safra tree nodes.

  Let \textsf{slice2safra} be the mapping which takes a ranked slice $(\alpha, t)$ and
  returns the ranked Safra tree given by the rank tree of $(\alpha, t)$, i.e. the tree
  structure defined by $\parent$ and the ordering of siblings given by the order of the
  corresponding sets in $t$.
\end{definition}

It is easy to see from the definitions that \textsf{safra2slice} and \textsf{slice2safra}
are injective and return a valid ranked slice and ranked Safra tree, respectively.
This implies that there exists a bijection between the sets of ranked Safra trees and ranked slices.
It is also not very hard to see that the following holds\withappendix{~(a proof can be found in Appendix~\ref{app:bijection})}:

\begin{restatable}{lemma}{lembijection}
  \label{lem:bijection}
  $\op{safra2slice}$ and $\op{slice2safra}$ are inverses of each other and provide a
  bijection between ranked Safra trees and ranked slices.
\end{restatable}

As we have established that both constructions, Muller-Schupp and Safra, operate on essentially the same structures,
from now on we talk about ranked slices and trees interchangeably.
Using this relationship, one can take the same tree/slice and apply both the successor
calculation of the Safra construction and of the Muller-Schupp construction to it. What one
first notices, is that the resulting tree/slice is very similar or equal in
many cases. This is owed to the fact that most operations in one construction have an
equivalent operation in the other, just formulated for the other representation.

For example, moving accepting successor states into a fresh child node in Safra's
construction corresponds to splitting accepting successors from non-accepting ones during
$\op{step}$ in the Muller-Schupp construction, as in the successor tuple the new child
(in the conceptual split-tree) gets a fresh, larger rank and by definition becomes the
rightmost child in the rank tree of the resulting new slice. The normalization steps that make the
successor sets pairwise disjoint also yield the same results. The ranks of nodes with green
events in the Safra construction coincide with ranks of sets that signal green in the
ranked slices, and ranks of Safra nodes with red events with ranks of sets that signal red.
The removal of empty sets by $\op{prune}$ and renumbering the ranks with $\op{normalize}$
is the same as the removal of the corresponding nodes in the Safra tree and updating
the LIR, i.e., the ranks of Safra nodes.

\begin{figure}
\begin{minipage}[t]{0.23\textwidth}
\begin{tikzpicture}[transform shape,baseline={([yshift=-.5ex]current bounding box.center)},shorten >=1pt,  inner sep=1pt,node distance=8mm,on grid,auto]
   \node[state,initial,initial text=] (q0)   {$q_0$};
   \node[state,accepting] (q1) [right=of q0] {$q_1$};
   \node[state,accepting] (q2) [below=of q0] {$q_2$};
   \node[state,accepting] (q3) [right=of q2] {$q_3$};
   \node[state] (q4) [right=of q1] {$q_4$};
   \node (lab) [left=of q0, yshift=12mm] {$\mc{A}:$};
    \path[->]
    (q0) edge [loop above] node {$a,b,c$} (q0)
    (q0) edge node  {$b$} (q1)
    (q0) edge node [swap] {$a$} (q2)
    (q2) edge [loop left] node {$b,c$} (q2)
            (q1) edge node {$c$} (q3)
    (q1) edge [bend left] node {$c$} (q4)
    (q4) edge [bend left] node {$b$} (q1)
    (q3) edge [loop right] node {$a,b$} (q3)
    (q4) edge [loop above] node {$a,c$} (q4)
    (q4) edge [bend left,near start] node {$c$} (q3);
\end{tikzpicture}
\end{minipage}
\begin{minipage}[t]{0.75\textwidth}
  \vspace{-1.3cm}
  ranked Safra tree sequence:\\[3mm]
\begin{tikzpicture}[transform shape,baseline={([yshift=-.5ex]current bounding box.center)},                    shorten >=1pt, inner sep=1pt,node distance=8mm,on grid,auto]
  \node (d0) [] {$\{q_0\}^1_\mcirc$};
  \node (d1) [below left=of d0] {$\{q_2\}^2_\triang$};
  \node (d2) [below right=of d0] {$\{q_4\}^3_\square$};
  \node (d3) [below=of d2,yshift=2mm] {$\{q_3\}^4_\hollowstar$};

  \node (de) [right=of d2] {$\overset{a}{\to}$};

  \node (e1) [right=of de] {$\{q_4\}^2_\square$};
  \node (e0) [above right=of e1] {$\{q_0\}^1_\mcirc$};
  \node (e2) [below right=of e0] {$\{q_2\}^4_\heartsuit$};
  \node (e3) [below=of e1,yshift=2mm] {$\{q_3\}^3_\hollowstar$};

  \node (ef) [right=of e2] {$\overset{c}{\to}$};

  \node (f1) [right=of ef] {$\{q_4\}^2_\square$};
  \node (f0) [above right=of f1] {$\{q_0\}^1_\mcirc$};
  \node (f2) [below right=of f0] {$\{q_2\}^3_\heartsuit$};
  \node (f3) [below=of f1,yshift=2mm] {$\{q_3\}^4_\diamondsuit$};

  \node (fg) [right=of f2] {$\overset{b}{\to}$};

  \node (g1) [right=of fg,xshift=1mm] {$\{q_1,q_3\}^2_\square$};
  \node (g0) [above right=of g1] {$\{q_0\}^1_\mcirc$};
  \node (g2) [below right=of g0,xshift=1mm] {$\{q_2\}^3_\heartsuit$};

  \path[]
  (d0) edge node {} (d1)
  (d0) edge node {} (d2)
  (d2) edge node {} (d3)
  (e0) edge node {} (e1)
  (e0) edge node {} (e2)
  (e1) edge node {} (e3)
  (f0) edge node {} (f1)
  (f0) edge node {} (f2)
  (f1) edge node {} (f3)
  (g0) edge node {} (g1)
  (g0) edge node {} (g2)
  ;
\end{tikzpicture}
\end{minipage}

\vspace{3mm}
\begin{minipage}[t]{0.49\textwidth}
\begin{tikzpicture}[transform shape,baseline={([yshift=-.5ex]current bounding box.center)},                    shorten >=1pt, inner sep=1pt,node distance=12mm,on grid,auto]
  \node (d0) [] {\( \{q_0\}^1_\mcirc \)};
  \node (d2) [left=of d0,xshift=-6mm] {\( \{q_4\}^3_\square \)};
  \node (d3) [left=of d2,xshift=2mm] {\( \{q_3\}^4_\hollowstar \)};
  \node (d1) [left=of d3,xshift=-0mm] {\( \{q_2\}^2_\triang \)};

  \node (txt) [above=6mm of d1,xshift=-3mm,anchor=west] {Muller-Schupp sequence of ranked slices:};

  \node (e0) [below right=of d0,xshift=-2mm,yshift=-0mm] {\( \{q_0\}^1_\mcirc \)};
  \node (e2) [below left=of d0,xshift=4mm,yshift=-0mm] {\( \{q_2\}^4_\heartsuit \)};
  \node (e1) [below right=of d2,xshift=-6mm,yshift=-0mm] {\( \{q_4\}^2_\square \)};
  \node (e3) [below left=of d3,xshift=2mm,yshift=-0mm] {\( \{q_3\}^3_\hollowstar \)};

  \node (f0) [below right=of e0,xshift=-6mm,yshift=-0mm] {\( \{q_0\}^1_\mcirc \)};
  \node (f2) [below left=of e2,xshift=6mm,yshift=-0mm] {\( \{q_2\}^3_\heartsuit \)};
  \node (f1) [below right=of e1,xshift=-14mm,yshift=-0mm] {\( \{q_4\}^2_\square \)};
  \node (f3) [below left=of e1,xshift=-10mm,yshift=-0mm] {\( \{q_3\}^4_\diamondsuit \)};

  \node (g0) [below right=of f0,xshift=-6mm,yshift=-0mm] {\( \{q_0\}^1_\mcirc \)};
  \node (g2) [below left=of f2,xshift=6mm,yshift=-0mm] {\( \{q_2\}^3_\heartsuit \)};
  \node (g1) [below right=of f1,xshift=-14mm,yshift=-0mm] {\( \{q_1\}^2_\square \)};
  \node (g3) [below left=of f1,xshift=-10mm,yshift=-0mm] {\( \{q_3\}^4_\diamondsuit \)};

  \node (de) [above right=of e0,xshift=-0mm,yshift=-3mm] {$\downarrow$ $a$};
  \node (ef) [above right=of f0,xshift=-0mm,yshift=-3mm] {$\downarrow$ $c$};
  \node (fg) [above right=of g0,xshift=-0mm,yshift=-3mm] {$\downarrow$ $b$};

  \path[->]
    (d2) edge[dotted] node {\tiny$\notin F$} (e1)
    (d3) edge [dotted,swap] node {\tiny$\in F$} (e3)
    (d0) edge [dotted,swap] node {\tiny$\in F$} (e2)
    (d0) edge[dotted] node {\tiny$\notin F$} (e0)
    (e0) edge[dotted] node {\tiny$\notin F$} (f0)
    (e2) edge [dotted,swap] node {\tiny$\in F$} (f2)
    (e1) edge [dotted,swap] node {\tiny$\in F$} (f3)
    (e1) edge[dotted] node [yshift=2mm,xshift=1mm] {\tiny$\notin F$} (f1)
    (f0) edge[dotted] node {\tiny$\notin F$} (g0)
    (f2) edge [dotted,swap] node {\tiny$\in F$} (g2)
    (f3) edge [dotted,swap] node {\tiny$\in F$} (g3)
    (f1) edge [dotted,swap] node {\tiny$\in F$} (g1)
    ;
\end{tikzpicture}
\end{minipage}
\hfill\vline\hfill
\begin{minipage}[t]{0.48\textwidth}
\begin{tikzpicture}[transform shape,baseline={([yshift=-.5ex]current bounding box.center)},                    shorten >=1pt, inner sep=1pt,node distance=8mm,on grid,auto]
  \node (q0) [] {\footnotesize $\{q_0\}^1$};
  \node (q1) [left=of q0] {\footnotesize $\emptyset^5$};
  \node (q2) [left=of q1,xshift=2mm] {\footnotesize $\emptyset^3$};
  \node (q3) [left=of q2] {\footnotesize $\{q_2\}^5$};
  \node (q4) [left=of q3] {\footnotesize $\mathbf{\emptyset^2}$};
  \node (q5) [left=of q4] {\footnotesize $\{q_1\}^5$};
  \node (q6) [left=of q5] {\footnotesize $\emptyset^4$};
  \node (q7) [left=of q6] {\footnotesize $\{q_3\}^5$};

  \node (txt) [above=13mm of q4] {During last transition, after $\op{step}$:};

  \node (ql) [left=of q7,xshift=3mm] {\footnotesize (};
  \node (qr) [right=of q0,xshift=-3mm] {\footnotesize )};

  \node (sl) [below=of ql,yshift=-10mm] {\footnotesize (};
  \node (s0) [right=of sl,xshift=-2mm] {\footnotesize $\{q_1,q_3\}^2$};
  \node (s1) [right=of s0,xshift=2mm] {\footnotesize $\{q_2\}^3$};
  \node (s2) [right=of s1,xshift=0mm] {\footnotesize $\{q_0\}^1$};
  \node (sr) [right=of s2,xshift=-4mm] {\footnotesize )};

  \node (mr) [yshift=-10mm,below=of qr] {\footnotesize )};
  \node (m3) [left=of mr,xshift=4mm] {\footnotesize $\{q_0\}^1$};
  \node (m2) [left=of m3] {\footnotesize $\{q_2\}^3$};
  \node (m1) [left=of m2] {\footnotesize $\{q_1\}^2$};
  \node (m0) [left=of m1] {\footnotesize $\{q_3\}^4$};
  \node (ml) [left=of m0,xshift=4mm] {\footnotesize (};

  \node (dul) [below=of q4,yshift=5mm,xshift=-3mm] {};
  \node (dur) [below=of q4,yshift=5mm,xshift=2mm] {};
  \node (ddl) [above=of s0,yshift=-4mm] {};
  \node (ddr) [above=of m1,yshift=-4mm] {};

    \path[->]
    (q0) edge [] node {} (q1)
    (q0) edge [bend right=40] node {} (q2)
    (q2) edge [] node {} (q3)
    (q0) edge [bend right=40] node {} (q4)
    (q4) edge [] node {} (q5)
    (q4) edge [bend right=40] node {} (q6)
    (q6) edge [] node {} (q7)

    (m3) edge [bend right=45] node {} (m2)
    (m3) edge [bend right=45] node {} (m1)
    (m1) edge [bend right=45] node {} (m0)

    (s2) edge [bend right=45] node {} (s1)
    (s2) edge [bend right=45] node {} (s0)
    ;
    \draw[-implies,double equal sign distance,swap] (dul) to node [yshift=1mm,xshift=-2mm] {\footnotesize Safra} (ddl);
    \draw[-implies,double equal sign distance] (dur) to node [yshift=0mm,xshift=2mm] {\footnotesize Muller-Schupp} (ddr);
\end{tikzpicture}
\end{minipage}

\caption{Transitions based on NBA $\mc{A}$ using both constructions.
The superscripts denote the ranks of tree nodes / sets in the slice tuple.
The subscripts are added for illustration purposes and
conceptually track nodes throughout time, i.e., the same symbol marks the ``same'' node at different times.
The algorithms agree on all but the last transition, where they differ due to
different handling of green nodes/ranks, in this case rank $2$ that marks an empty set
after calculating and splitting the successors (illustrated on the bottom right).
In the Muller-Schupp case, the rank is moved left during $\op{prune}$, resulting in a
child being pulled into the parent in the rank tree, whereas in the Safra construction the
whole subtree is collapsed. The solid edges between sets depict the rank tree induced by
$\parent$, dotted edges depict the edges in the conceptual split-tree. In the bottom
right the slices are shown together with their tree interpretation.
}
\vspace{-3mm}
\label{fig:saframullerschupp}
\end{figure}
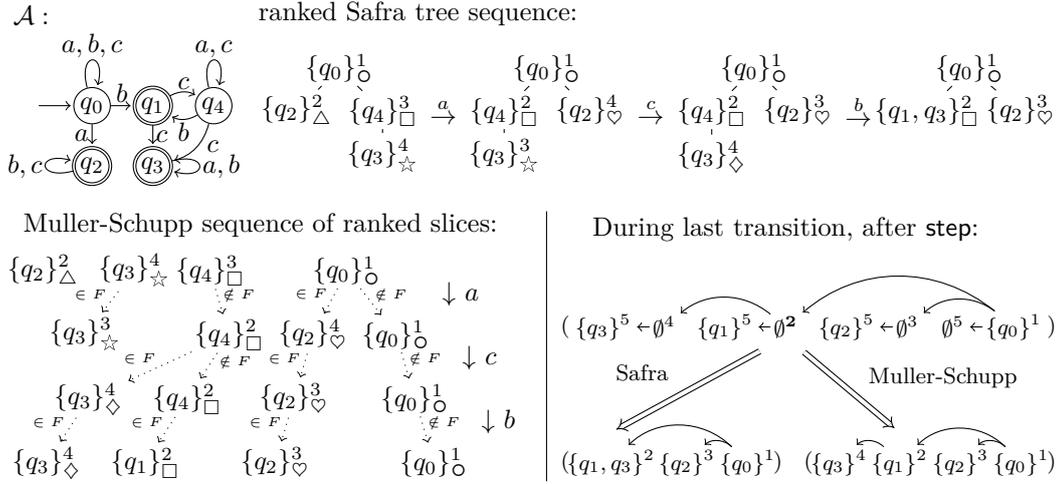

In fact, \emph{the only difference} between the constructions is what
happens with a tree node in case of a green event. Recall that in
Safra's construction, the whole subtree of a green node is collapsed
to a single node.  In the Muller-Schupp construction, the green ranks
are those that end up on an empty set after $\op{step}$, and that
survive the $\op{prune}$ operation, in which the ranks are moved to
the next non-empty set to the left, and only the minimal ones are kept
on each non-empty set. In the view of ranked trees, this corresponds
to a green node absorbing its rightmost, uppermost child node into it,
while keeping the rest of the subtree unchanged. See
Figure~\ref{fig:saframullerschupp} for an illustration.

After observing that both constructions differ in only a minor step and noticing that both
yield correct (but possibly different) automata, it becomes apparent that the exact step
performed for green events is not essential and there must be a more general
mechanism to uncover. The construction we present in Section~\ref{sec:construction} results from this
line of thought.

On the practical side, it is worth mentioning that the cost of switching between the
representations using the presented bijection is negligible---the traversal of a ranked
Safra tree to obtain a ranked slice is obviously possible in linear time.  For the other
direction there also exists a simple linear time algorithm\withappendix{~(presented in Appendix~\ref{app:lintime})}
that calculates the parent and left subtree boundary relation from the ranking $\alpha$.
 \section{The unified construction}
\label{sec:construction}

In this section, we present a construction that builds on the
Muller-Schupp construction from Section~\ref{sec:mullerschupp}, and
unifies it with Safra's construction by adding another operation,
called $\op{merge}$, between $\op{prune}$ and $\op{normalize}$:
$(\alpha,t) \xrightarrow{\op{step}}
(\hat{\alpha}, \hat{t})  \xrightarrow{\op{prune}}
(\tilde{\alpha}, \tilde{t}) \xrightarrow{\op{merge}}
(\check\alpha, \check{t}) \xrightarrow{\op{normalize}}
(\alpha',t')
$.
This new operation is nondeterministic, and can be instantiated in
different ways. In particular, it can be instantiated trivially and
thus corresponds to the Muller-Schupp construction, and it can be used
to emulate the Safra construction.

We first describe the idea of $\op{merge}$, and then give a formal
definition. Assume that, after $\op{step}$ and $\op{prune}$ have been
applied to some ranked slice $(\alpha, t)$, we have the pre-slice
$(\tilde{\alpha}, \tilde{t})$, and the dominating (minimal active) rank
$k$ (determined by $\op{prune}$, see Section~\ref{sec:mullerschupp}).  Then
$\op{merge}$ can collapse groups of neighbouring sets in the tuple,
and preserves the minimum rank from each collapsed range, similar to
$\op{prune}$. In contrast to $\op{prune}$, which ``merges'' one
non-empty set with multiple empty sets in a deterministic manner,
$\op{merge}$ may actually take the union of multiple adjacent
non-empty sets, depending on the ranks currently assigned to them.

The non-overlapping intervals of sets that are collapsed together are
not uniquely determined in general.
They only have to satisfy the constraints that no sets with rank
smaller than the dominating rank $k$ are merged with anything
else, and that the set with rank $k$ is not merged with anything to the
right of it. These constraints are important for the correctness, and
ensure that in the ranked Safra tree perspective, the nodes with rank
smaller than $k$ do not change, and that the node with the dominating
rank $k$ is not merged with sets outside of its subtree.

Formally, $\op{merge}$ returns a pre-slice $(\check\alpha, \check{t})$
obtained in the following way (see Figure~\ref{fig:generalmerge} for
an illustration).
Let $I_1, I_2, \ldots, I_{n'}$ be a sequence of sets partitioning the set of indices
$\{1,\ldots, \tilde{n}\}$ in $\tilde{t}$ into adjacent groups, i.e., $\min I_1 = 1$, $\max
I_{n'} = \tilde{n}$ and for all $j>1$ we have $\min I_j = \max I_{j-1} + 1$.
This grouping should satisfy the following property for all $1\leq j \leq n'$ and $l \in
I_j$: if $\tilde\alpha(l) < k$, then $|I_j| = 1$, and if $\tilde\alpha(l) = k$, then
$\max I_j = l$. Then the pre-slice $(\check\alpha, \check{t})$ is
defined by the sets
$\check{S}_i := \bigcup_{j\in I_i} \tilde{S}_j$ and the ranking
function $\check\alpha(i) :=
\min\{\tilde{\alpha}(j) \mid j \in I_i \}$ for all $i\in\{1,\ldots,n'\}$, i.e., for each interval, the union of the sets and
the smallest rank is taken.
\begin{figure}
  \begin{center}
    \begin{tikzpicture}[baseline={([yshift=-.5ex]current bounding box.center)},
      shorten >=1pt,node distance=1cm,inner sep=1pt,on grid,auto]
      \node   (t)                 {$(\tilde{\alpha}, \tilde{t})$};
      \node   (eq)  [right of=t]  {$= ($};

      \def\atlabs{{">k",">k","<k",">k",">k","k",">k","<k"}}
      \node   (s1) [right of=eq]    {${\tilde{S}_1}$};
      \node   (a1) [above right= 1.1mm and 4mm of s1]
                     {$\scriptstyle \pgfmathparse{\atlabs[1]}\pgfmathresult$};
      \foreach \x in {2,...,8}{
        \pgfmathtruncatemacro{\prev}{\x - 1}
        \node   (s\x) [right of=s\prev,xshift=4mm]   {${\tilde{S}_\x}$};
        \node   (a\x) [above right= 1.1mm and 4mm of s\x]
                        {$\scriptstyle \pgfmathparse{\atlabs[\prev]}\pgfmathresult$};
      }
      \node     (cl)  [right of=s8, xshift=3mm]  {$)$};

      \node   (t2)  [below of=t, yshift=-08mm]                {$(\check{\alpha}, \check{t})$};
      \node   (eq2) [right of=t2]  {$= ($};
      \node   (cl2)  [below of=cl, yshift=-08mm]  {$)$};

      \node (u1) [below of=s1,xshift=2mm] {$\bigcup$};
      \node (m1) [below of=s2,xshift=-2mm] {$\min$};

      \node (u3) [below of=s5,xshift=-4mm] {$\bigcup$};
      \node (m3) [below of=s5,xshift=6mm] {$\min$};

      \path (u1) -- (m1) node[midway,yshift=-10mm] (s12) {${\check{S}_1}$};
      \node   (s22) [below of=s3,yshift=-8mm]    {${\check{S}_2}$};
      \path (u3) -- (m3) node[midway,yshift=-10mm] (s32) {${\check{S}_3}$};
      \node   (s42) [below of=s7,yshift=-8mm]    {${\check{S}_4}$};
      \node   (s52) [below of=s8,yshift=-8mm]    {${\check{S}_5}$};
      \foreach \x in {1,...,5}{
        \node   (a\x2) [above right= 1.1mm and 5mm of s\x2] {$\scriptstyle \check\alpha(\x)$};
      }

      \path[->] (t) edge node [xshift=1mm] {$\op{merge}$} (t2);

      \foreach \x in {2,3,6,7}{
        \pgfmathtruncatemacro{\succ}{\x + 1}
        \draw[black!30,thick,dashed] ([yshift=2mm] $(s\x)!0.65!(s\succ)$)
                                  -- ([yshift=-2cm]$(s\x)!0.65!(s\succ)$);
      }

      \draw[decorate,decoration={brace}]                  ([yshift=3mm] $(eq)!0.5!(s1)$)
        -- node[above=1mm] {$\scriptstyle I_1=\{1,2\}$}   ([yshift=3mm] $(s2)!0.5!(s3)$);
      \draw[decorate,decoration={brace}]                  ([yshift=3mm] $(s2)!0.75!(s3)$)
        -- node[above=1mm] {$\scriptstyle I_2=\{3\}$}     ([yshift=3mm] $(s3)!0.5!(s4)$);
      \draw[decorate,decoration={brace}]                  ([yshift=3mm] $(s3)!0.75!(s4)$)
        -- node[above=1mm] {$\scriptstyle I_3=\{4,5,6\}$} ([yshift=3mm] $(s6)!0.5!(s7)$);
      \draw[decorate,decoration={brace}]                  ([yshift=3mm] $(s6)!0.75!(s7)$)
        -- node[above=1mm] {$\scriptstyle I_4=\{7\}$}     ([yshift=3mm] $(s7)!0.5!(s8)$);
      \draw[decorate,decoration={brace}]                  ([yshift=3mm] $(s7)!0.75!(s8)$)
        -- node[above=1mm] {$\scriptstyle I_5=\{8\}$}     ([yshift=3mm] $(s8)!0.5!(cl)$);

      \path[->,dotted]
        (s1) edge (u1)
        (s2) edge (u1)
        (u1) edge (s12)

        (s4) edge (u3)
        (s5) edge (u3)
        (s6) edge (u3)
        (u3) edge (s32)

        (s3) edge (s22)
        (s7) edge (s42)
        (s8) edge (s52)
        ;
      \path[->]
        (a1) edge (m1)
        (a2) edge (m1)
        (m1) edge (a12)

        (a4) edge (m3)
        (a5) edge (m3)
        (a6) edge[bend left] (m3)
        (m3) edge (a32)

        (a3) edge (a22)
        (a7) edge (a42)
        (a8) edge (a52)
        ;
    \end{tikzpicture}
  \end{center}
  \caption{
    Illustration of the general $\op{merge}$ operation that comes after $\op{prune}$ and before
    $\op{normalize}$, with the minimal active rank $k$ and ranks depicted as set
    superscripts. The illustrated intervals
        are the coarsest partitioning of indices in $\tilde{t}$ satisfying the constraints.
  }
  \label{fig:generalmerge}
\end{figure}
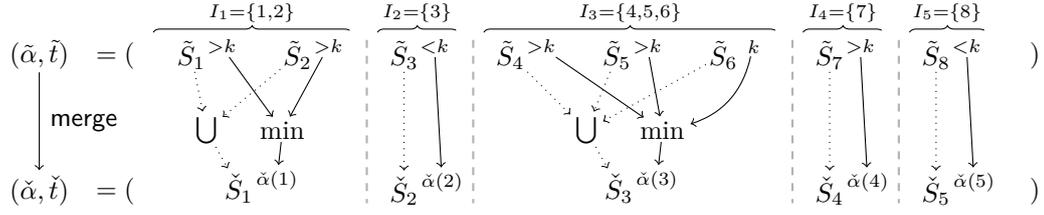

As in the Muller-Schupp construction, $\op{normalize}$ is applied to $(\check\alpha, \check{t})$ to obtain the
successor macrostate $(\alpha', t')$. This extended transition relation is used to
obtain the transition-based deterministic parity automaton, as before.

An example showing how the choice of different merge strategies leads to different successor
states is illustrated in Figure~\ref{fig:concreteex}.
Observe that we can recover the Muller-Schupp construction by using the identity function
for $\op{merge}$, or in other words, putting each index into its own interval, which is the finest partitioning of indices that satisfies
the requirements on $\op{merge}$.
On the other hand, we can also take the coarsest
compatible partitioning, i.e., minimize
the number of intervals. We call this kind of update \emph{maximal collapse}.

We can emulate a Safra-update by imposing some additional constraints
on the intervals, ensuring that only the complete subtrees of nodes
with green ranks are merged. More concretely, we require that
intervals that are not singletons span exactly the nodes of the
complete subtree that is rooted in a green rank in the view of the
slice as ranked Safra tree. Note that for an index $\ell$ in the
tuple, the subtree of the corresponding node in the ranked Safra tree
corresponds to the interval that starts one step right of the left
subtree boundary of $\ell$, and ends in $\ell$, that is, the interval
$\lsibling(l) +1, \ldots,\ell$.
Thus, for imitating the Safra merge rule, the intervals $I_1, I_2,
\ldots, I_{n'}$ from $\op{merge}$ are the unique smallest
intervals satisfying
\begin{gather*}
  \forall i\in [n'],l\in I_i:
    \quad \tilde\alpha(l) \in G \implies \lsibling(l) + 1 \in I_i
    \quad\text{(complete subtrees collapsed)}.
        \end{gather*}

\begin{proposition}
\label{prp:emulate-safra}
The operation $\op{merge}$ can be instantiated such that the
transitions of the constructed TDPA correspond to the transitions of
the Muller-Schupp construction or to the transitions of the Safra construction.
\end{proposition}

\begin{figure}[!htbp]
    \begin{center}
    $\mc{A}:$
    \begin{tikzpicture}[baseline={([yshift=-.5ex]current bounding box.center)},
      shorten >=1pt,node distance=1cm,inner sep=1pt,on grid,auto]
      \node[state] (q0) {$q_0$};
      \node[state] (q1) [below right=of q0] {$q_1$};
      \node[state,accepting] (q5) [below left=of q0] {$q_5$};
      \node[state,accepting] (q3) [below=of q0,yshift=-1.3cm] {$q_3$};
      \node[state,accepting] (q2) [above right=of q3] {$q_2$};
      \node (q4)[state] [above left=of q3] {$q_4$};

      \path[->]
        (q0) edge[black!30] node {} (q1)
        (q0) edge[black!30] node {} (q5)
        (q0) edge[black!30,bend left] node {} (q4)
        (q4) edge[black!30] node {} (q5)
        (q4) edge[black!30] node {} (q3)
        (q3) edge[black!30,bend left] node {} (q1)
        (q3) edge[black!30] node {} (q2)
        (q2) edge[] node {} (q1)

        (q0) edge[loop left] node {} (q0)
        (q1) edge[black!30,loop right] node {} (q1)
                (q2) edge[loop right] node {} (q2)
        (q3) edge[loop right] node {} (q3)
        (q4) edge[loop left] node {} (q4)
        (q5) edge[loop left] node {} (q5)
        ;
    \end{tikzpicture}
      \hspace{1cm}
      $(\alpha, t) = $
    \begin{tikzpicture}[baseline={([yshift=-.5ex]current bounding box.center)},
      shorten >=1pt,node distance=1cm,inner sep=1pt,on grid,auto]
      \node (n1) {$\{q_0\}^1$};
      \node (n2) [below left =of n1] {$\{q_1\}^2$};
      \node (n3) [below right=of n1] {$\{q_4\}^4$};
      \node (n4) [below right=of n3] {$\{q_5\}^5$};
      \node (n5) [below left=of n2] {$\{q_2\}^3$};
      \node (n6) [below right=of n2] {$\{q_3\}^6$};

      \path (n1) edge (n2) (n1) edge (n3) (n3) edge (n4) (n2) edge (n5) (n2) edge (n6);
    \end{tikzpicture}

    \vspace{3mm}
    \begin{tabular}{c}
    \begin{tikzpicture}[baseline={([yshift=-.5ex]current bounding box.center)},
      shorten >=1pt,node distance=1cm,inner sep=1pt,on grid,auto]
      \node[]   (t)                 {$(\alpha, t)$};
      \node[]   (eq)  [right of=t]  {$= ($};

      \node[]   (s1) [right of=eq]    {$\{q_2\}^3$};
      \node[]   (s2) [right of=s1, xshift=2mm]   {$\{q_3\}^5$};
      \node[]   (s3) [right of=s2, xshift=2mm]   {$\{q_1\}^2$};
      \node[]   (s4) [right of=s3, xshift=2mm]   {$\{q_5\}^6$};
      \node[]   (s5) [right of=s4, xshift=2mm]   {$\{q_4\}^4$};
      \node[]   (s6) [right of=s5, xshift=2mm]   {$\{q_0\}^1$};

      \node[]   (cl)  [right of=s6,xshift=-2mm]  {$)$};

      \node[]   (t2) [below of=t, yshift=-05mm]                {$(\tilde{\alpha}, \tilde{t})$};
      \node[]   (eq2)  [right of=t2]  {$= ($};

      \node[]   (s12) [right of=eq2]               {$\{q_2\}^7$};
      \node[]   (s22) [right of=s12, xshift=2mm]   {$\{q_1\}^3$};
      \node[]   (s32) [right of=s22, xshift=2mm]   {$\{q_3\}^2$};
      \node[]   (s42) [right of=s32, xshift=2mm]   {$\{q_5\}^6$};
      \node[]   (s52) [right of=s42, xshift=2mm]   {$\{q_4\}^4$};
      \node[]   (s62) [right of=s52, xshift=2mm]   {$\{q_0\}^1$};

      \node[]   (cl2)  [below of=cl,yshift=-05mm]  {$)$};

      \path[->]
        (t) edge node[xshift=1mm,align=center] {$\op{step};$\\$\op{prune}$} (t2)
        (s1) edge node[near end] {$\scriptstyle \in F$} (s12)
        (s1) edge node[near end] {$\scriptstyle \in \overline{F}$} (s22)
        (s2) edge node[near end] {$\scriptstyle \in F$} (s32)
        (s4) edge node[] {$\scriptstyle \in F$} (s42)
        (s5) edge node[] {$\scriptstyle \in \overline{F}$} (s52)
        (s6) edge node[] {$\scriptstyle \in \overline{F}$} (s62)
        ;
    \end{tikzpicture}
      \\
    $\scriptstyle G = \{2,6\}\quad R = \{5\} \quad\implies\quad A = \{2,5,6\} \implies \mathbf{k = 2}$
    \end{tabular}

    \vspace{10mm}
    \begin{tabular}{lll}
      M.-S.: &
      \begin{tikzpicture}[baseline={([yshift=-.5ex]current bounding box.center)},
        shorten >=1pt,node distance=9mm and 1cm,inner sep=1pt,on grid,auto]
        \node[]   (t)                 {$(\tilde{\alpha}, \tilde{t})$};
        \node[]   (eq)  [right of=t]  {$= ($};

        \node[]   (s1) [right of=eq,xshift=-3mm]    {$\{q_2\}^7$};
        \node[]   (s2) [right of=s1, xshift=2mm]   {$\{q_3\}^3$};
        \node[]   (s3) [right of=s2, xshift=2mm]   {$\{q_1\}^2$};
        \node[]   (s4) [right of=s3, xshift=2mm]   {$\{q_5\}^6$};
        \node[]   (s5) [right of=s4, xshift=2mm]   {$\{q_4\}^4$};
        \node[]   (s6) [right of=s5, xshift=2mm]   {$\{q_0\}^1$};

        \node[]   (cl)  [right of=s6,xshift=-4mm]  {$)$};

        \node[]   (t2) [below of=t, yshift=-03mm]                {$({\alpha'}, {t'})$};
        \node[]   (eq2)  [right of=t2]  {$= ($};

        \node[]   (s12) [right of=eq2,xshift=-3mm]               {$\{q_2\}^6$};
        \node[]   (s22) [right of=s12, xshift=2mm]   {$\{q_1\}^3$};
        \node[]   (s32) [right of=s22, xshift=2mm]   {$\{q_3\}^2$};
        \node[]   (s42) [right of=s32, xshift=2mm]   {$\{q_5\}^5$};
        \node[]   (s52) [right of=s42, xshift=2mm]   {$\{q_4\}^4$};
        \node[]   (s62) [right of=s52, xshift=2mm]   {$\{q_0\}^1$};

        \node[]   (cl2)  [below of=cl,yshift=-03mm]  {$)$};

        \path[->]
          (t) edge node[xshift=1mm,align=center] {$\op{merge};$\\$\op{norm.}$} (t2)
          ;

        \foreach \x in {1,2,3,4,5}{
          \pgfmathtruncatemacro{\succ}{\x + 1}
          \draw[black!30,thick,dashed] ([yshift=2mm] $(s\x)!0.5!(s\succ)$)
                                    -- ([yshift=-1.7cm]$(s\x)!0.5!(s\succ)$);
        }
      \end{tikzpicture}
      &
      \begin{tikzpicture}[baseline={([yshift=-.5ex]current bounding box.center)},
        shorten >=1pt,node distance=1cm,inner sep=1pt,on grid,auto]
        \node (n1) {$\{q_0\}^1$};
        \node (n2) [below left =of n1,yshift=2mm] {$\{q_3\}^2$};
        \node (n3) [below right=of n1,yshift=2mm] {$\{q_4\}^4$};
        \node (n4) [below=of n3,yshift=4mm] {$\{q_5\}^5$};
        \node (n5) [below left=of n2,yshift=2mm,xshift=3mm] {$\{q_1\}^3$};
        \node (n6) [below left=of n5,yshift=2mm,xshift=3mm] {$\{q_2\}^6$};

        \path (n1) edge (n2) (n1) edge (n3) (n3) edge (n4) (n2) edge (n5) (n5) edge (n6);
      \end{tikzpicture}

      \\[15mm]
      Safra: &
      \begin{tikzpicture}[baseline={([yshift=-.5ex]current bounding box.center)},
        shorten >=1pt,node distance=9mm and 1cm,inner sep=1pt,on grid,auto]
        \node[]   (t)                 {$(\tilde{\alpha}, \tilde{t})$};
        \node[]   (eq)  [right of=t]  {$= ($};

        \node[]   (s1) [right of=eq,xshift=-3mm]    {$\{q_2\}^7$};
        \node[]   (s2) [right of=s1, xshift=2mm]   {$\{q_3\}^3$};
        \node[]   (s3) [right of=s2, xshift=2mm]   {$\{q_1\}^2$};
        \node[]   (s4) [right of=s3, xshift=2mm]   {$\{q_5\}^6$};
        \node[]   (s5) [right of=s4, xshift=2mm]   {$\{q_4\}^4$};
        \node[]   (s6) [right of=s5, xshift=2mm]   {$\{q_0\}^1$};

        \node[]   (cl)  [right of=s6,xshift=-4mm]  {$)$};

        \node[]   (t2) [below of=t, yshift=-03mm]                {$({\alpha'}, {t'})$};
        \node[]   (eq2)  [right of=t2]  {$= ($};

        \node[]   (s12) [right of=eq2,xshift=-3mm]               {$\{$};
        \node[]   (s22) [right of=s12, xshift=2mm]   {$q_1, q_2, q_3$};
        \node[]   (s32) [right of=s22, xshift=2mm]   {$\}^2$};
        \node[]   (s42) [right of=s32, xshift=2mm]   {$\{q_5\}^4$};
        \node[]   (s52) [right of=s42, xshift=2mm]   {$\{q_4\}^3$};
        \node[]   (s62) [right of=s52, xshift=2mm]   {$\{q_0\}^1$};

        \node[]   (cl2)  [below of=cl,yshift=-03mm]  {$)$};

        \path[->]
          (t) edge node[xshift=1mm,align=center] {$\op{merge};$\\$\op{norm.}$} (t2)
          ;

        \foreach \x in {3,4,5}{
          \pgfmathtruncatemacro{\succ}{\x + 1}
          \draw[black!30,thick,dashed] ([yshift=2mm] $(s\x)!0.5!(s\succ)$)
                                    -- ([yshift=-1.7cm]$(s\x)!0.5!(s\succ)$);
        }
      \end{tikzpicture}
      &
      \begin{tikzpicture}[baseline={([yshift=-.5ex]current bounding box.center)},
        shorten >=1pt,node distance=1cm,inner sep=1pt,on grid,auto]
        \node (n1) {$\{q_0\}^1$};
        \node (n2) [below left =of n1,yshift=2mm,xshift=-2mm] {$\{q_1,q_2,q_3\}^2$};
        \node (n3) [below right=of n1,yshift=2mm,xshift=2mm] {$\{q_4\}^3$};
        \node (n4) [below=of n3,yshift=4mm] {$\{q_5\}^4$};
        \path (n1) edge (n2) (n1) edge (n3) (n3) edge (n4);
      \end{tikzpicture}

      \\[15mm]
      Max.: &
      \begin{tikzpicture}[baseline={([yshift=-.5ex]current bounding box.center)},
        shorten >=1pt,node distance=9mm and 1cm,inner sep=1pt,on grid,auto]
        \node[]   (t)                 {$(\tilde{\alpha}, \tilde{t})$};
        \node[]   (eq)  [right of=t]  {$= ($};

        \node[]   (s1) [right of=eq,xshift=-3mm]    {$\{q_2\}^7$};
        \node[]   (s2) [right of=s1, xshift=2mm]   {$\{q_3\}^3$};
        \node[]   (s3) [right of=s2, xshift=2mm]   {$\{q_1\}^2$};
        \node[]   (s4) [right of=s3, xshift=2mm]   {$\{q_5\}^6$};
        \node[]   (s5) [right of=s4, xshift=2mm]   {$\{q_4\}^4$};
        \node[]   (s6) [right of=s5, xshift=2mm]   {$\{q_0\}^1$};

        \node[]   (cl)  [right of=s6,xshift=-4mm]  {$)$};

        \node[]   (t2) [below of=t, yshift=-03mm]                {$({\alpha'}, {t'})$};
        \node[]   (eq2)  [right of=t2]  {$= ($};

        \node[]   (s12) [right of=eq2,xshift=-3mm]               {$\{$};
        \node[]   (s22) [right of=s12, xshift=2mm]   {$q_1, q_2, q_3$};
        \node[]   (s32) [right of=s22, xshift=2mm]   {$\}^2$};
        \node[]   (s42) [right of=s32, xshift=2mm]   {$\{q_4,$};
        \node[]   (s52) [right of=s42, xshift=2mm]   {$q_5\}^3$};
        \node[]   (s62) [right of=s52, xshift=2mm]   {$\{q_0\}^1$};

        \node[]   (cl2)  [below of=cl,yshift=-03mm]  {$)$};

        \path[->]
          (t) edge node[xshift=1mm,align=center] {$\op{merge};$\\$\op{norm.}$} (t2)
          ;

        \foreach \x in {3,5}{
          \pgfmathtruncatemacro{\succ}{\x + 1}
          \draw[black!30,thick,dashed] ([yshift=2mm] $(s\x)!0.5!(s\succ)$)
                                    -- ([yshift=-1.7cm]$(s\x)!0.5!(s\succ)$);
        }
      \end{tikzpicture}
      &
      \begin{tikzpicture}[baseline={([yshift=-.5ex]current bounding box.center)},
        shorten >=1pt,node distance=1cm,inner sep=1pt,on grid,auto]
        \node (n1) {$\{q_0\}^1$};
        \node (n2) [below left =of n1,xshift=-1mm] {$\{q_1,q_2,q_3\}^2$};
        \node (n3) [below right=of n1,xshift=1mm] {$\{q_4,q_5\}^3$};
        \path (n1) edge (n2) (n1) edge (n3);
      \end{tikzpicture}

    \end{tabular}

    \vspace{5mm}
      \begin{tikzpicture}[baseline={([yshift=-.5ex]current bounding box.center)},
        shorten >=1pt,node distance=1cm,inner sep=1pt,on grid,auto]
        \node (n1) {$\{q_0\}^1$};
        \node (n2) [below left =of n1,yshift=2mm] {$\{q_3\}^2$};
        \node (n3) [below right=of n1,yshift=2mm,xshift=-2mm] {$\{q_4,q_5\}^4$};
                \node (n5) [below left=of n2,yshift=2mm,xshift=3mm] {$\{q_1\}^3$};
        \node (n6) [below left=of n5,yshift=2mm,xshift=3mm] {$\{q_2\}^5$};

        \path (n1) edge (n2) (n1) edge (n3) (n2) edge (n5) (n5) edge (n6);
      \end{tikzpicture}
      \begin{tikzpicture}[baseline={([yshift=-.5ex]current bounding box.center)},
        shorten >=1pt,node distance=1cm,inner sep=1pt,on grid,auto]
        \node (n1) {$\{q_0\}^1$};
        \node (n2) [below left =of n1,yshift=0mm] {$\{q_3\}^2$};
        \node (n3) [below right=of n1,yshift=0mm,xshift=-2mm] {$\{q_4,q_5\}^4$};
                \node (n5) [below=of n2,yshift=4mm] {$\{q_1,q_2\}^3$};

        \path (n1) edge (n2) (n1) edge (n3) (n2) edge (n5);
      \end{tikzpicture}
      \begin{tikzpicture}[baseline={([yshift=-.5ex]current bounding box.center)},
        shorten >=1pt,node distance=1cm,inner sep=1pt,on grid,auto]
        \node (n1) {$\{q_0\}^1$};
        \node (n2) [below left =of n1,yshift=0mm] {$\{q_1,q_3\}^2$};
        \node (n3) [below right=of n1,yshift=0mm] {$\{q_4,q_5\}^3$};
                \node (n5) [below=of n2,yshift=4mm] {$\{q_2\}^4$};

        \path (n1) edge (n2) (n1) edge (n3) (n2) edge (n5);
      \end{tikzpicture}
      \begin{tikzpicture}[baseline={([yshift=-.5ex]current bounding box.center)},
        shorten >=1pt,node distance=1cm,inner sep=1pt,on grid,auto]
        \node (n1) {$\{q_0\}^1$};
        \node (n2) [below left =of n1,yshift=0mm] {$\{q_3\}^2$};
        \node (n3) [below right=of n1,yshift=0mm,xshift=-2mm] {$\{q_4\}^4$};
        \node (n4) [below=of n3,yshift=4mm] {$\{q_5\}^5$};
        \node (n5) [below=of n2,yshift=4mm] {$\{q_1,q_2\}^3$};

        \path (n1) edge (n2) (n1) edge (n3) (n2) edge (n5) (n3) edge (n4);
      \end{tikzpicture}
      \begin{tikzpicture}[baseline={([yshift=-.5ex]current bounding box.center)},
        shorten >=1pt,node distance=1cm,inner sep=1pt,on grid,auto]
        \node (n1) {$\{q_0\}^1$};
        \node (n2) [below left =of n1,yshift=0mm] {$\{q_1,q_3\}^2$};
        \node (n3) [below right=of n1,yshift=0mm] {$\{q_4\}^3$};
        \node (n4) [below=of n3,yshift=4mm] {$\{q_5\}^4$};
        \node (n5) [below=of n2,yshift=4mm] {$\{q_2\}^5$};

        \path (n1) edge (n2) (n1) edge (n3) (n2) edge (n5) (n3) edge (n4);
      \end{tikzpicture}

  \end{center}

  \caption{$\mc{A}$ illustrates the relevant part of an NBA during a transition on some
    symbol $x\in \Sigma$, that is, the arrows correspond to the $x$-transitions of $\mc{A}$.
    The gray edges are the ones pruned in the reduced transition relation $\Delta_t$.
   The current macrostate $(\alpha,t)$ is represented as the rank tree to the right of $\mc{A}$, and as ranked slice below $\mc{A}$.
  The $\op{step}$ and $\op{prune}$ operations
  (see Fig.~\ref{fig:mullerschupp} for details) result in ranks 1,3 and
  4 being passed down along the right child.   Ranks 2 and 6 are moved to the left    and hence are green. Rank 5 is overwritten by 2 and hence is red. Rank 7 is a fresh
  rank which is larger than the others. The dominating rank $k$ is 2.
  The choice of different merge intervals (as shown in Fig.~\ref{fig:generalmerge})
  results in different successors.
  The successors for the three discussed variants, Muller-Schupp, Safra, and maximal collapse, are shown as rank trees on the right.
  The 5 other permitted successors are depicted at the bottom.
    }
  \label{fig:concreteex}
\end{figure}
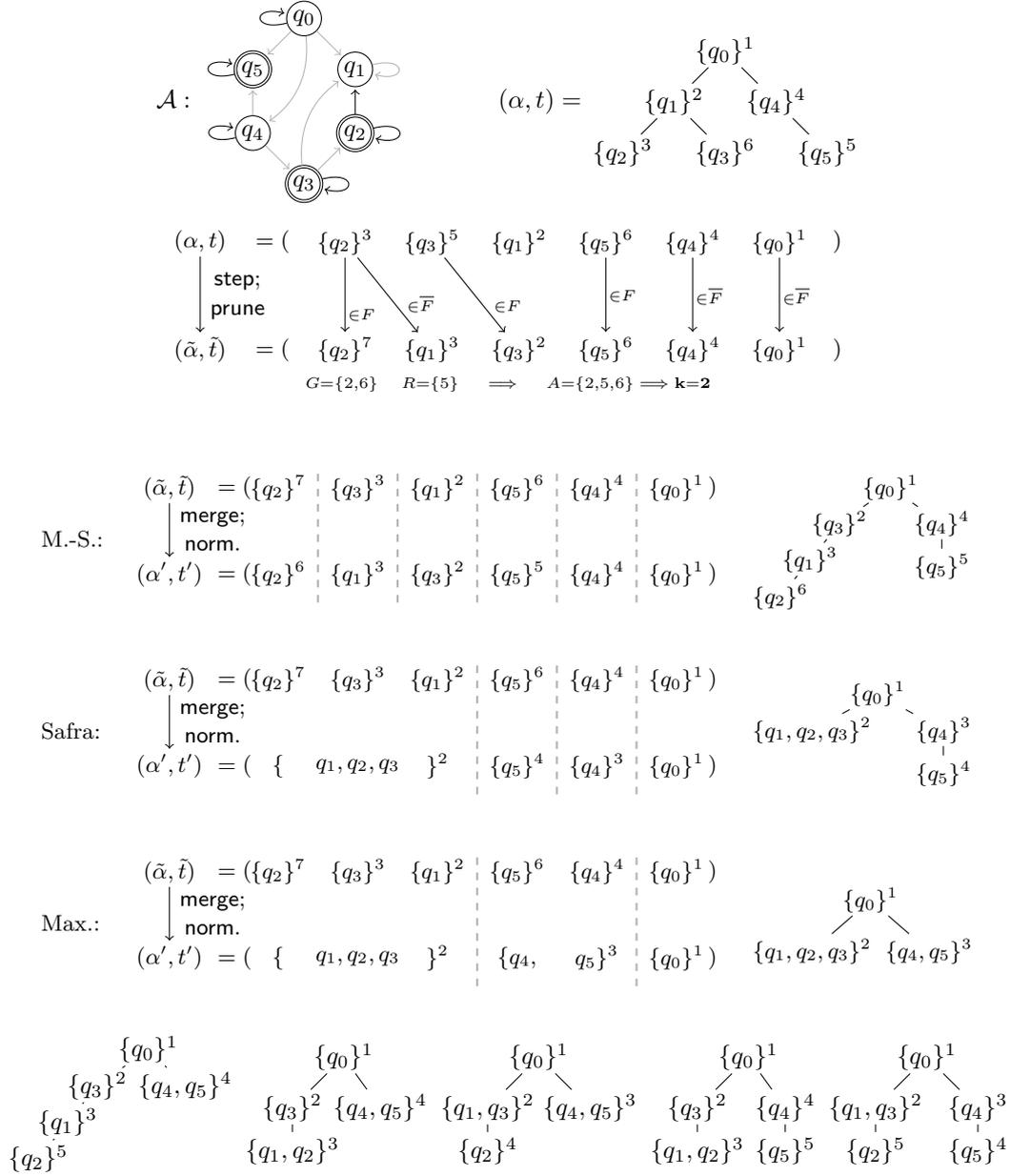

Notice that for all merge rules except for the Muller-Schupp update, the relationship of
ranked slices and consecutive levels of the reduced split-tree (see Section~\ref{sec:mullerschupp}) breaks
down. One can, however, reflect the merges also in the reduced
split-tree by doing the merges of the corresponding sets on each
level, which leads to an acyclic graph instead of a tree. This view is helpful in the correctness proof of the construction.

\begin{restatable}{theorem}{thmcorrectness}
  \label{thm:correctness}
  Let $\mc{A}$ be an NBA. Then a deterministic parity automaton $\mc{B}$,
  obtained by the described determinization construction, has at most $\mc{O}(n!^2)$
  states and recognizes the same language as $\mc{A}$.
\end{restatable}

The upper bound holds because the same macrostates are used as in the presented
Muller-Schupp construction in Section~\ref{sec:mullerschupp}. The correctness can be shown by a refinement of
the original correctness proof of the Safra construction \cite{safra1988complexity}\withappendix{, and is given in  Appendix~\ref{app:correctness}}.
 \section{Discussion and Conclusion}
\label{sec:discussion}

We have presented a new variant of the Muller-Schupp construction for
determinization of Büchi automata into parity automata, reducing the
information stored in the macrostates to ordered tuples of disjoint
sets annotated with ranks. These ranked slices are in bijection with
ranked Safra trees, which leads to a general construction that can
emulate the Muller-Schupp construction and the
Safra-construction. This answers, in some sense, the question from
\cite{fogarty2015profile} on the relation between the two types of
constructions.

In general, one can obtain many different valid deterministic automata by
choosing different deterministic transition functions that are compatible with the
described successor relation. One can also imagine this as constructing a non-deterministic
automaton with \emph{all} permitted successors, and then pruning the edges arbitrarily,
while preserving for each state only one outgoing transition for each symbol, to ``carve''
out a valid deterministic automaton.

This non-determinism comes from two sources.  One degree of freedom in
our construction comes from the different ways of assigning ranks (to
new nodes, and when closing gaps resulting from deleted ranks). This
freedom is already mentioned in \cite{schewe2009tighter}. But here the
flexibility is just in the choice of the specific permutation, which
still describes structurally the same tree in any case. The novel and
in our opinion powerful degree of freedom in our construction is the
possibility for different valid $\op{merge}$ operations, which allows
for a vastly larger pool of possible successors, as the results may
describe structurally different trees. Furthermore, the smaller the
smallest active rank, the more different a permitted successor may
look like.

We have explicitly mentioned the merge strategies that lead to the
Muller-Schupp and Safra constructions, and also have mentioned a third
strategy, the maximal collapse rule that merges as many sets as
possible
(as shown in e.g. Figure~\ref{fig:concreteex}).
We also want to point out that, while
fixing one such $\op{merge}$-rule for the whole construction is the
simplest implementation, the construction permits using \emph{any}
valid successor without the need to disambiguate the $\op{merge}$
operation beforehand, i.e., picking the successor of a state from the
set of permitted ones is a \emph{local} choice. One may think of
schemes where the successor is chosen dynamically, depending on the
input or already computed information. For example, one can check
whether a valid successor has already been constructed, and only add a
new state according to a fixed policy if this is not the case. We have
already implemented a prototype making use of such an optimization
(among others) with encouraging results.

We also want to point out that the presented construction works equally well with
transition-based Büchi automata as input, in which case the $\op{step}$ operation
separates states which are reached by at least one accepting transition from those that are
not. One can easily verify that this does not impact the reasoning in the proofs.

It is also possible to adapt the construction to yield Rabin automata,
such that the corresponding Safra construction as presented in
\cite{schewe2009tighter} is subsumed. In this setting, however, the
presentation of macrostates as ordered tuples of sets is less
natural. Furthermore, in this setting the merges of sets needs to be
restricted to subtrees of green nodes, because there is no total order
of importance of nodes as provided by the ranks.

\bibliography{literature}

\begin{thebibliography}{10}

\bibitem{baier2008principles}
Christel Baier and Joost-Pieter Katoen.
\newblock {\em Principles of model checking}.
\newblock MIT Press, 2008.

\bibitem{buchi1966symposium}
J~Richard B{\"u}chi.
\newblock On a decision method in restricted second order arithmetic.
\newblock In {\em Studies in Logic and the Foundations of Mathematics},
  volume~44, pages 1--11. Elsevier, 1966.

\bibitem{ColcombetZ09}
Thomas Colcombet and Konrad Zdanowski.
\newblock A tight lower bound for determinization of transition labeled
  b{\"u}chi automata.
\newblock In {\em ICALP 2009}, volume 5556 of {\em Lecture Notes in Computer
  Science}, pages 151--162. Springer, 2009.

\bibitem{fogarty2015profile}
Seth Fogarty, Orna Kupferman, Moshe~Y Vardi, and Thomas Wilke.
\newblock Profile trees for b{\"u}chi word automata, with application to
  determinization.
\newblock {\em Information and Computation}, 245:136--151, 2015.

\bibitem{GastinO01}
Paul Gastin and Denis Oddoux.
\newblock Fast {LTL} to b{\"{u}}chi automata translation.
\newblock In {\em Computer Aided Verification, 13th International Conference,
  {CAV} 2001, Paris, France, July 18-22, 2001, Proceedings}, pages 53--65,
  2001.
\newblock \href {http://dx.doi.org/10.1007/3-540-44585-4\_6}
  {\path{doi:10.1007/3-540-44585-4\_6}}.

\bibitem{GiannakopoulouH01}
Dimitra Giannakopoulou and Klaus Havelund.
\newblock Automata-based verification of temporal properties on running
  programs.
\newblock In {\em 16th {IEEE} International Conference on Automated Software
  Engineering {(ASE} 2001), 26-29 November 2001, Coronado Island, San Diego,
  CA, {USA}}, pages 412--416, 2001.
\newblock \href {http://dx.doi.org/10.1109/ASE.2001.989841}
  {\path{doi:10.1109/ASE.2001.989841}}.

\bibitem{kahler2008complementation}
Detlef K{\"a}hler and Thomas Wilke.
\newblock Complementation, disambiguation, and determinization of b{\"u}chi
  automata unified.
\newblock In {\em ICALP 2008}, pages 724--735. Springer, 2008.

\bibitem{mcnaughton1966testing}
Robert McNaughton.
\newblock Testing and generating infinite sequences by a finite automaton.
\newblock {\em Information and control}, 9(5):521--530, 1966.

\bibitem{MeyerSL18}
Philipp~J. Meyer, Salomon Sickert, and Michael Luttenberger.
\newblock Strix: Explicit reactive synthesis strikes back!
\newblock In {\em Computer Aided Verification - 30th International Conference,
  {CAV} 2018, Held as Part of the Federated Logic Conference, FloC 2018,
  Oxford, UK, July 14-17, 2018, Proceedings, Part {I}}, pages 578--586, 2018.
\newblock \href {http://dx.doi.org/10.1007/978-3-319-96145-3\_31}
  {\path{doi:10.1007/978-3-319-96145-3\_31}}.

\bibitem{muller1995simulating}
David~E Muller and Paul~E Schupp.
\newblock Simulating alternating tree automata by nondeterministic automata:
  New results and new proofs of the theorems of rabin, mcnaughton and safra.
\newblock {\em Theoretical Computer Science}, 141(1-2):69--107, 1995.

\bibitem{piterman2006nondeterministic}
Nir Piterman.
\newblock From nondeterministic buchi and streett automata to deterministic
  parity automata.
\newblock In {\em Logic in Computer Science, 2006 21st Annual IEEE Symposium
  on}, pages 255--264. IEEE, 2006.

\bibitem{redziejowski2012improved}
Roman~R Redziejowski.
\newblock An improved construction of deterministic omega-automaton using
  derivatives.
\newblock {\em Fundamenta Informaticae}, 119(3-4):393--406, 2012.

\bibitem{safra1988complexity}
Shmuel Safra.
\newblock On the complexity of omega-automata.
\newblock In {\em 29th Annual Symposium on Foundations of Computer Science,
  1988.}, pages 319--327. IEEE, 1988.

\bibitem{schewe2009tighter}
Sven Schewe.
\newblock Tighter bounds for the determinisation of b{\"u}chi automata.
\newblock In {\em FOSSACS}, pages 167--181. Springer, 2009.

\bibitem{SomenziB00}
Fabio Somenzi and Roderick Bloem.
\newblock Efficient {B}{\"{u}}chi automata from {LTL} formulae.
\newblock In {\em Computer Aided Verification, 12th International Conference,
  {CAV} 2000, Chicago, IL, USA, July 15-19, 2000, Proceedings}, pages 248--263,
  2000.
\newblock \href {http://dx.doi.org/10.1007/10722167\_21}
  {\path{doi:10.1007/10722167\_21}}.

\bibitem{Thomas97}
Wolfgang Thomas.
\newblock Languages, automata, and logic.
\newblock In Grzegorz Rozenberg and Arto Salomaa, editors, {\em Handbook of
  Formal Languages, Vol. 3}, pages 389--455. Springer-Verlag New York, Inc.,
  New York, NY, USA, 1997.

\bibitem{Thomas08}
Wolfgang Thomas.
\newblock Church's problem and a tour through automata theory.
\newblock In {\em Pillars of Computer Science, Essays Dedicated to Boris (Boaz)
  Trakhtenbrot on the Occasion of His 85th Birthday}, pages 635--655. Springer,
  2008.
\newblock \href {http://dx.doi.org/10.1007/978-3-540-78127-1}
  {\path{doi:10.1007/978-3-540-78127-1}}.

\bibitem{VardiW07}
Moshe~Y. Vardi and Thomas Wilke.
\newblock Automata: from logics to algorithms.
\newblock In {\em Logic and automata - history and perspectives}, volume~2 of
  {\em Texts in Logic and Games}, pages 629--724. Amsterdam University Press,
  2007.

\end{thebibliography}

\withappendix{
\clearpage
\appendix
\section{Proof of the bijection between ranked Safra trees and ranked slices}
\label{app:bijection}

In this section we provide the missing proof for Lemma~\ref{lem:bijection}.
We make use of the following technical lemma to omit one direction of the proof, using the
property that both functions are already known to be injective:
\begin{lemma}
  \label{lem:injinv}
  Let $f : A \to B, g : B \to A$ be two injective functions and $g \circ f = \op{id}_A$.
  Then it follows that $f$ and $g$ are inverse mappings of a bijection.
\end{lemma}
\begin{proof}
  From injectivity of $g$ follows that $g$ has a left inverse $g^{li}$. Hence:
  \[
    g^{li} \circ g = \op{id}_B
    \Rightarrow g^{li} \circ \overbrace{g \circ f}^{= \op{id}_A} \circ \overbrace{g \circ
    f}^{= \op{id}_A} = \op{id}_B \circ f \circ\overbrace{g \circ f}^{= \op{id}_A}
    \Rightarrow g^{li} = f
  \]

  As both are left inverses of each other, they are also right inverses of each other by
  definition and therefore $f^{-1} = g$ and $g^{-1} = f$.
\end{proof}

Now we can show the actual statement:
{\renewcommand{\thetheorem}{\ref{lem:bijection}}
\lembijection*
\addtocounter{theorem}{-1}
}

\begin{proof}
  As both mappings are injective, by Lemma~\ref{lem:injinv} it suffices to argue that
  $\op{slice2safra}\circ \mathsf{safra2slice}$ returns the original ranked Safra tree to
  prove the statement.

  Take a ranked Safra tree with $n$ nodes.
    Recall that a parent node always has a smaller rank than its children, and siblings to
  the right of a node have larger ranks than that node.

  Construct the ranked slice $(\alpha, t)$ using $\op{safra2slice}$ and then consider its
  rank tree. Clearly, both trees have the same number of nodes and the same set of labels.
  Also, nodes with the same label also have the same rank.
    What remains to be shown is the equality of the tree structure. Let $\op{node}(S_i)$
  denote the tree node in the original ranked Safra tree that is labelled by set $S_i$
  from $t$.

  Notice that by the visiting order of the post-order traversal it is ensured that
  parent nodes, which have smaller assigned ranks, appear to the right of all their
  descendants in the tuple and thus closer ancestors of a node are listed earlier than
  further ancestors. Thus, the set $S_j$ of the parent $\op{node}(S_j)$ of some node $\op{node}(S_i)$
  will be the closest set to the right of the set $S_i$ in the tuple $t$ which has a smaller rank.
  Hence, by definition of $\parent$, the original tree is reconstructed.
  For the ordering of siblings, observe that if $\op{node}(S_i)$ is the left sibling of
  $\op{node}(S_j)$, then $\op{node}(S_j)$ is visited only after completing the subtree of
  $\op{node}(S_i)$ and hence $S_j$ appears later in the tuple.
    Hence, the rank tree obtained via $\uparrow$ from the ranked slice is exactly the
  original ranked Safra tree.

\end{proof}

 \newpage
\section{From ranked slices to rank trees in linear time}
\label{app:lintime}

It is easy to see that going from a ranked Safra tree to a ranked slice is possible in
linear time, as $\op{safra2slice}$ is just a depth-first traversal. In this section we
show how to efficiently compute $\op{slice2safra}$, i.e., how the rank tree can be
obtained in linear time from a ranked slice, which can be useful in implementations of the
presented determinization construction.

The parent and left subtree boundary relationships, which capture the tree structure of the ranked
slice, can be computed using the following algorithm:

\begin{algorithm}
\begin{algorithmic}
  \Function{unflatten}{$\alpha, n$} \Comment{ranking $\alpha:[n]\to[n]$, $n$ sets in tuple}
    \State $P:= \text{new array}, L:= \text{new array}, S:= \text{new stack}$ \Comment{parent, left sub., parent stack}
            \For{$i:= n \text{\textbf{ down to }} 1$} \Comment{main loop}
      \While{$S \text{ not empty and } \alpha(i) < \alpha(\op{top}(S))$} \Comment{inner loop}
        \State{$\op{L}(\op{top}(S)) := i, \op{pop}(S)$}
      \EndWhile
      \IIf{$S$ not empty}{$\op{P}(i) := \op{top}(S)$}
      \State $\op{push}(S,i)$
          \EndFor

    \While{$S \text{ not empty}$} \Comment{clean-up loop}
      \State{$\op{L}(\op{top}(S)) := 0,\op{pop}(S)$}
    \EndWhile

    \hspace{-2mm}\Return $(P,L)$
  \EndFunction
\end{algorithmic}
\end{algorithm}

\begin{lemma}
  $\textsc{unflatten}$ calculates the rank tree from a ranked slice $(\alpha, t)$ in linear time.
\end{lemma}
\begin{proof}
  Let $(\alpha, t)$ be a ranked slice with $t=(S_1, \ldots, S_n)$.
  The main loop is repeated exactly $n$ times. In each iteration one index is
  pushed onto the stack. All other loops are iterated at most $n$ times in total, as they
  require a non-empty stack and pop an element in each iteration. Hence, the algorithm
  completes in linear time. It remains to be shown that the resulting arrays $P$ and $L$
  correspond to $\parent$ and $\lsibling$.

  In the first iteration of the main loop the stack is empty and hence only the index of
  the last set $S_n$ with rank $1$ (i.e., the root) is pushed. This index is never popped
  from the stack during the main loop, as is has the smallest rank, and in the last
  iteration of the clean-up loop is assigned $0$ as left subtree boundary index and has an undefined
  parent. Therefore, the root is treated correctly.

  Next we analyze under which conditions
  the assignments to $P$ and $L$ can be incorrect and show that they lead to a contradiction.
  Let $j$ and $k$ with $1\leq j<k\leq n$ be two indices. Notice that index $k$ is pushed onto the stack
  before $j$ by the main loop.

  First consider the case that $\parent(j) = k$.

  For contradiction, assume that $k$ is popped from the stack before the loop iteration with $i=j$
  ($i$ referring to the variable in the algorithm). This implies that
  there is an index $l$ with $j < l < k$, $\alpha(l) < \alpha(k)$ and $k$ being on top
  of the stack in the loop iteration with $i=l$.
  But as $\alpha(j) > \alpha(l) > \alpha(k)$ violates the assumption that $\parent(j) = k$
  by definition of $\parent$, because then $l$ should be the parent of $j$.

  Next, assume that $k$ is not on top of the stack in the iteration with $i=j$. Then there
  is an $l$ with $j<l<k$ and $\alpha(l) > \alpha(k)$   which is pushed onto the stack and stays there until $P(j)$ is assigned. As $\alpha(l) > \alpha(j)$,
  $j$ must be in a different subtree to the left of $l$. But then index $l$ must be
  removed in the inner loop during the iteration $i=j$, before $P(j)$ is assigned, which
  is a contradiction. Hence, we conclude that the parent array is assigned correctly.

  Now, consider the case that $j=\lsibling(k)$.

  First assume that $j=0$. Then there is no index with a smaller rank to the left
  of $k$ and hence the inner loop was not entered whenever $k$ was on top of the stack, so
  that $k$ remained on the stack during the main loop. But then $k$ will be assigned $0$
  as left subtree boundary in the clean-up loop.

  Now assume that $j>0$, i.e., there exists a non-trivial left subtree boundary.
  Again, assume for contradiction that $k$ is popped before loop iteration with $i=j$.
  Then there was an index
  $l$ with $j < l < k$ such that $k$ was on top of the stack and $\alpha(l) < \alpha(k)$,
  so that $k$ was popped from the stack in the inner loop. By definition of
  $\lsibling$ this implies that $l$ is the left subtree boundary of $k$, violating the assumption.

  Finally, assume that $k$ is not on top of the stack when $i = j$. Then some index $l$ with
  $j<l<k$ was pushed onto the stack with $\alpha(l) > \alpha(k)$ and in the main loop
  iteration with $i=j$ the last such index $l$ was assigned as parent of $j$, because
  $\alpha(j) > \alpha(l)$ (implied by the fact that $l$ remained on the stack).
  But then $\alpha(j) > \alpha(k)$, contradicting the assumption that $j$ is the left
  subtree boundary index of $k$, which must have a smaller rank.
      Hence we conclude that the left subtree boundary must also be assigned correctly,
  completing the proof.

\end{proof}
 \newpage
\section{Correctness proof for the determinization construction}
\label{app:correctness}

In this section, we provide the proof for the main result:

\thmcorrectness*

As discussed, the claim on the state complexity directly follows from
the upper bound given in \cite[Proposition~2]{schewe2009tighter}, and the bijection
between the set of ranked slices and the set of ranked Safra trees presented in
Section~\ref{sec:bijection}.

Hence, it remains to be shown that an automaton $\mc{B}$ with transitions that satisfy the
described successor relation of the determinization construction presented in
Section~\ref{sec:construction} indeed recognizes the same language as $\mc{A}$.

For the rest of this section, fix some arbitrary word $w\in \Sigma^\omega$, NBA $\mc{A}$
and let $\mc{B}$ be the obtained TDPA.

\subsection{From an accepting NBA run to an accepting TDPA run}

To show that $\mc{B}$ accepts $w$ if there exists an accepting run in $\mc{A}$, we  introduce
the concept of the \emph{rank-profile} of a state in a pre-slice $(\alpha,t)$, which is just the sequence of nodes in
the corresponding ranked Safra tree starting from the root, down to the node labelled by the set containing the
state. The sequence is strictly ascending, because children have larger ranks
than their parents. Formally, the \emph{rank-profile} $\op{rp} : Q_t \to [n]^*$ maps each
state $q$ to a sequence $a_1a_2\ldots a_m$ such that $a_1=\alpha(n)=1$, $a_m = \alpha(q)$,
and $a_{i-1} = \op{\uparrow_\alpha}(a_i)$ is the rank of the parent
of the node with rank $a_i$, for all $i>1$.
Consider, for example, the ranked slice
$(\{q_3\}^4,\{q_1\}^2,\{q_2\}^3,\{q_0\}^1)$. The rank profile of $q_3$
is $1,2,4$, and the rank profile of $q_2$ is $1,3$.
Notice that rank-profiles have at
most the length $|Q|$ and are fully determined by the ascending set of
ranks, so we can interpret the profile also as a set of ranks.

Let $a=a_1\ldots a_{n_a},b=b_1\ldots b_{n_b}$ be two rank profiles with $m:=\min(n_a,n_b)$
being the length of the shorter one. We say that $a$ is \emph{better} than $b$
and write $a \prec b$, if $a_1\ldots a_m < b_1\ldots b_m$ wrt.\
lexicographic order, or if $a_1\ldots a_m=b_1\ldots b_m$ and $n_a >
n_b$. So the ``better than'' relation almost corresponds to ``being
smaller in the
lexicographic ordering'' but with the difference that
in our ordering a strict prefix of another rank profile is not better
but worse.
This definition is useful because of the following property:

\begin{lemma}
  \label{lem:rp_idx}
  Let $(\alpha,t)$ be a ranked slice, and $p,q\in Q_t$.

  Then $\op{idx}(p) < \op{idx}(q) \Longleftrightarrow \op{rp}(p) \prec \op{rp}(q)$.
\end{lemma}
\begin{proof}
  If $\op{idx}(p) < \op{idx}(q)$, then, by definition of the rank tree, either $p$ is in the
  subtree of $\alpha(q)$ or in the subtree of some left sibling of $q$ or one of its
  ancestors. In the first
  case, the rank profile of $p$ must agree with the rank profile of $q$ on the whole
  length of the latter and then, as $p$ is a descendant, there must be at least one additional node
  along the branch leading to $p$ from $q$, i.e., the rank profile is longer. Hence,
  $\op{rp}(p) \prec \op{rp}(q)$. In the second case, the rank profiles of $p$ and $q$ must
  agree on some prefix up to their latest common ancestor, whereas afterwards the rank
  profile of $p$ continues with a node at the root of some different subtree, which is
  located more to the left in the rank tree and has a smaller rank.
  But then the rank profile is lexicographically smaller and again we have $\op{rp}(p) \prec \op{rp}(q)$.
    Clearly, reversing this case analysis shows the other direction.

\end{proof}

The \emph{$k$-cut} of a rank-profile $a$ is a prefix of the rank-profile including all
ranks less than $k$ and, if possible, the first rank which is larger than or equal to $k$.
Formally, $\cut{k}(a_1\ldots a_n) := a_1\ldots a_i$ with $i := \min(\{n\}\cup\{i \mid a_i\geq k\})$.
Consider, for example, the rank profile $1,2,4$. For $k \ge 3$, the $k$-cut is the
complete profile $1,2,4$, the $2$-cut is $1,2$, and the $1$-cut is $1$.
The ordering $\preceq_k$ compares the $k$-cuts of two rank profiles, i.e., for rank
profiles $a$ and $b$, $a \preceq_k b$ iff $\cut{k}(a) \preceq \cut{k}(b)$ and consequently,
$a =_k b$ iff $a \preceq_k b$ and $b\preceq_k a$. For example, $1,2,4 =_2 1,2,5$, but
$1,2,4 \prec_3 1,2,5$.

We write $\cut{k}(q)$ for $\cut{k}(\op{rp}(q))$.
This notion is useful when we analyze how the rank profile can change over time in
case that the smallest rank that is active infinitely often is $k$.

In the presentation of the construction in Sections~\ref{sec:mullerschupp} and
\ref{sec:construction} it was specified that the rank assigned to the left child sets
during the $\op{step}$ operation is the same fresh rank, and only after $\op{normalize}$ all sets have
different ranks. To simplify the reasoning in the proofs, in the following we assume that
the left child sets all get fresh and unused, but already \emph{different} ranks. In this case,
$\op{normalize}$ only closes ``gaps'' due to removed ranks.  It is easy to
see, that both formulations are equivalent, but this variant has the advantage that
the rank-tree (and also rank-profiles) is also defined on all intermediate pre-slices and
that $\op{normalize}$ never makes the rank assigned to a set larger than before.

We start by showing that the constructed DPA $\mc{B}$ accepts all
words that are accepted by the NBA $\mc{A}$.

\begin{lemma}
  \label{lem:nba_to_dpa}
  $w\in L(\mc{A}) \implies w \in L(\mc{B})$
\end{lemma}
\begin{proof}
  Let $\rho=q_0q_1\ldots$ be an accepting run of $\mc{A}$ on $w$.
  Take the run of $\mc{B}$ on $w$ and consider its sequence of ranked slices
  $s_0,s_1,\ldots$ with $s_i = (\alpha_i, t_i)$.
    Let $p_0, p_1, \ldots$ be the sequence of rank-profiles such that
  $p_i=\op{rp}(q_i)$ are the rank profiles of states along the run,
            let $k_0, k_1, \ldots$ be the sequence of dominating ranks, where $k_i$ is the
  dominating rank in the transition from $s_i$ to $s_{i+1}$,
  let $k$ be the smallest $k_i$ appearing infinitely often in this sequence,
  and finally, let $c_i := \cut{k_i}(p_i)$ and $m_i := \max c_i$ be the $k_i$-cuts of the rank
  profiles and the last rank along those prefixes, respectively.
    The proof is structured as follows:
  \begin{enumerate}
    \item First, we show that $p_i \succeq_{k_i} p_{i+1}$ holds in every transition.
    \item This implies that for some time $i_0$, $p_i =_k p_{i+1}$ must hold for all $i>i_0$.
    \item We show that the last ranks $m_i$ of the $k_i$-cuts $c_i$ must be $\geq k$ for all $i>i_0$.
    \item Finally, we show that rank $k$ cannot be red after $i_0$ and thus must be inf.\ often green.
  \end{enumerate}

  \textbf{1}:
    Pick some time $i$ and consider the transition from $i$ to $i+1$.

  Notice that regardless of the current dominating rank $k_i$, the rank profile $p_i$ cannot
  get worse (i.e., increase) during $\op{normalize}$ or $\op{step}$.
    For $\op{normalize}$ this is easy to see, because by definition it must preserve the
  relative ordering and only closes the unused rank ``gaps'' after removal of empty sets
  and eventual merges. Hence it can make a rank profile only lexicographically smaller in
  each position by reassigning ranks, without changing the length.
    For $\op{step}$, observe that if $q_{i+1} \in \Delta_{t_i}(q_i, w(i+1))$, then $\op{step}$ either
  puts $q_{i+1}$ into a set with the same rank as before (right branch) or, if $q_{i+1}$
  is accepting (left branch), moves it into a child wrt. the rank tree, thereby increasing
  the length of the rank profile. If $q_{i+1} \not\in \Delta_{t_i}(q_i, w(i+1))$, then, by
  definition of the restricted transition relation $\Delta_{t_i}$, $q_{i+1}$ must have a
  predecessor $q_i'$ in a set with lower tuple index. By Lemma~\ref{lem:rp_idx} this
  implies that $\op{rp}_{\hat{t}_{i}}(q_{i+1}) \preceq \op{rp}_{t_i}(q_i') \prec
  \op{rp}_{t_i}(q_i)$. Therefore, the \op{step} operation also can only make the rank
  profile better.

  As we have the dominating rank $k_i$, it means that no rank $<k_i$ was red or green,
  which means for the rank profile $p_i$ that no rank $<k_i$ marked an empty set after
  $\op{step}$ or was overwritten during $\op{prune}$. Also, $\op{merge}$ may not modify
  sets with ranks $<k_i$. Hence, if $m_i < k_i$, then clearly the $k_i$-cut $c_i$ is not
  influenced in the transition.

  If $m_i \geq k_i$, notice that regardless whether $k_i$ was red or green, at least one rank
  was removed during $\op{prune}$. If $k_i$ itself was red, then all sets with ranks $>k_i$ will be
  reassigned a smaller rank to close the gap. If $k_i$ was green, it means that it has
  overwritten some larger rank during $\op{prune}$. If the overwritten rank was $m_i$,
  then the $k_i$-cut decreased directly. Otherwise, some other rank was
  overwritten and the set with rank $m_i$ either keeps the same rank during
  $\op{normalize}$ (if $m_i$ was smaller than the overwritten rank), or it will be
  reassigned to smaller rank (if it was larger). Neither $\op{prune}$ nor $\op{merge}$ can
  make the $k_i$-cut shorter than before, as for $\op{prune}$ this would imply that a rank
  $<k_i$ was green and for $\op{merge}$, that it collapses sets with ranks $<k_i$, which
  is forbidden.

  \textbf{2}:
  There is some time $i'$ after which no rank $<k$ is dominating again. After that, we
  have $p_i \succeq_k p_{i+1}$ for all $i > i'$, as $p_i \succeq_{k'} p_{i+1}$ implies
  that $p_i \succeq_k p_{i+1}$ for all $k'>k$ (if one sequence is not lexicographically
  larger than the other, clearly this also applies to their prefixes). Hence, eventually,
  after some $i_0 > i'$ we have that $p_i =_k p_{i+1}$ for all $i > i_0$, i.e., the
  $k$-cut prefix eventually stabilizes and the rank profile then can only change at
  positions after this prefix.

  \textbf{3}:
  Next, we show that $m_i \geq k$ for all $i>i_0$. For contradiction, assume that $m_i < k$ at
  some time $i>i_0$. Observe that the accepting run $\rho$ visits an accepting state $q_F$
  infinitely often. This means, that eventually $q_F$ goes into the left child set during
  $\op{step}$, which would either make $m_i$ (which is  $< k$ by assumption) green, if the
  right child set is empty, or otherwise it would make the $k$-cut longer, violating either the fact
  that no rank $<k$ is active after $i_0$, or that the $k$-cut does not change anymore
  after $i_0$.

  \textbf{4}:
  Finally, observe that $k$ cannot be red after $i_0$. If $k$ was red after $i_0$, it means
  that it was either overwritten by some smaller green rank, which we excluded by choice
  of $k$ and $i_0$, or it would mark an empty set after $\op{step}$ and be removed during
  prune, which is also not possible, because if a rank is removed not due to being
  overwritten, it means that all sets to the left were empty, whereas some set must
  contain the current state of the run $\rho$, and it cannot be to the right of $k$
  because $m_i \geq k$.

  As $k$ is active infinitely often and cannot be red after $i_0$, it can only be red
  finitely often and must be green infinitely often, which implies that the smallest
  assigned priority is even and $\mc{B}$ accepts $w$.

\end{proof}

\subsection{From an accepting TDPA run to an accepting NBA run}

To show the other direction, it is convenient to define the notion of the run-DAG.
Intuitively, the run-DAG can be obtained from the
reduced split-tree (see Section~\ref{sec:mullerschupp}) by applying the corresponding
$\op{merge}$ operation on each level of the tree (i.e., merging the corresponding sets and
redirecting the edges to the new union set), before constructing the next level.

\begin{definition}
  Let $w\in \Sigma^\omega$ and $s_0, s_1, \ldots$ with $s_i = (\alpha_i, t_i)$ the
  sequence of macrostates on $w$.
    The \emph{run-DAG} of $\mc{B}$ on $w$ is defined as follows. Level $i$ of
  the DAG has the sets $S_{i,1}, \ldots, S_{i,n}$ of the tuple $t_i$ as nodes. An edge
  $S_{i,j} \to S'_{i+1,j'}$ exists between sets on two adjacent levels $i$ and $i+1$ of
  the DAG, if the target set contains a non-trivial subset of normalized successors of the
  source set, or formally, if $\emptyset \neq \Delta_{t_i}(S_{i,j}, w(i)) \subseteq
  S'_{i+1,j'}$. The edge is \emph{marked} (written $S_{i,j} \overset{M}{\to}
  S'_{i+1,j'}$), if $\Delta_{t_i}(S_{i,j})\cap S'_{i+1,j'} \subseteq F$,
  i.e., if it contains no non-accepting normalized successors from the source set.
  Otherwise, it is \emph{unmarked}.
  An infinite path in the run-DAG is a \emph{bad path}, if it contains only unmarked edges.
\end{definition}

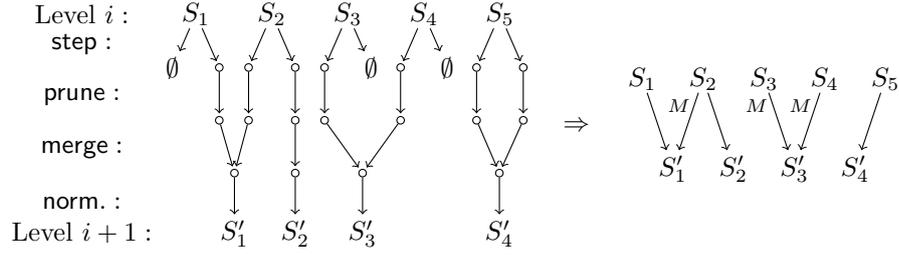
\begin{figure}[htbp]
  \begin{center}
    \begin{tikzpicture}[baseline={([yshift=-.5ex]current bounding box.center)},
      shorten >=1pt,node distance=1cm,inner sep=1pt,on grid,auto]
      \node   (l1)                 {Level $i:$};
      \node[right=of l1,xshift=5mm]   (s1)                 {$S_1$};
      \node[right=of s1]   (s2)                 {$S_2$};
      \node[right=of s2]   (s3)                 {$S_3$};
      \node[right=of s3]   (s4)                 {$S_4$};
      \node[right=of s4]   (s5)                 {$S_5$};

      \node[below=of l1,yshift=6mm]   (l2)                 {$\op{step}:$};

      \node[below left=of s1,xshift=4mm]           (st1)        {$\emptyset$};
      \node[state,below right=of s1,xshift=-4mm]   (st2)        {};
      \node[state,below left=of s2,xshift=4mm]     (st3)        {};
      \node[state,below right=of s2,xshift=-4mm]   (st4)        {};
      \node[state,below left=of s3,xshift=4mm]     (st5)        {};
      \node[below right=of s3,xshift=-4mm]         (st6)        {$\emptyset$};
      \node[state,below left=of s4,xshift=4mm]     (st7)        {};
      \node[below right=of s4,xshift=-4mm]         (st8)        {$\emptyset$};
      \node[state,below left=of s5,xshift=4mm]     (st9)        {};
      \node[state,below right=of s5,xshift=-4mm]   (st10)       {};

      \node[below=of l2,yshift=3mm]   (l3)                 {$\op{prune}:$};

      \node[state,below=of st2,yshift=3mm]   (pr1)        {};
      \node[state,below=of st3,yshift=3mm]   (pr2)        {};
      \node[state,below=of st4,yshift=3mm]   (pr3)        {};
      \node[state,below=of st5,yshift=3mm]   (pr4)        {};
      \node[state,below=of st7,yshift=3mm]   (pr5)        {};
      \node[state,below=of st9,yshift=3mm]   (pr6)        {};
      \node[state,below=of st10,yshift=3mm]  (pr7)        {};

      \node[below=of l3,yshift=3mm]   (l4)                 {$\op{merge}:$};

      \node[state,below=of pr1,yshift=3mm,xshift=2mm]   (me1)        {};
      \node[state,below=of pr3,yshift=3mm]   (me2)        {};
      \node[state,below=of pr4,yshift=3mm,xshift=5mm]   (me3)        {};
      \node[state,below=of pr6,yshift=3mm,xshift=3mm]   (me4)        {};

      \node[below=of l4,yshift=3mm]   (l5)                 {$\op{norm.}:$};

      \node[below=of me1,yshift=2mm]   (n1)        {$S'_1$};
      \node[below=of me2,yshift=2mm]   (n2)        {$S'_2$};
      \node[below=of me3,yshift=2mm]   (n3)        {$S'_3$};
      \node[below=of me4,yshift=2mm]   (n4)        {$S'_4$};

      \node[below=of l5,yshift=6mm]   (l6)                 {Level $i+1:$};

      \path[->]
        (s1) edge (st1)
        (s1) edge (st2)
        (s2) edge (st3)
        (s2) edge (st4)
        (s3) edge (st5)
        (s3) edge (st6)
        (s4) edge (st7)
        (s4) edge (st8)
        (s5) edge (st9)
        (s5) edge (st10)
        (st2) edge (pr1)
        (st3) edge (pr2)
        (st4) edge (pr3)
        (st5) edge (pr4)
        (st7) edge (pr5)
        (st9) edge (pr6)
        (st10) edge (pr7)
        (pr1) edge (me1)
        (pr2) edge (me1)
        (pr3) edge (me2)
        (pr4) edge (me3)
        (pr5) edge (me3)
        (pr6) edge (me4)
        (pr7) edge (me4)
        (me1) edge (n1)
        (me2) edge (n2)
        (me3) edge (n3)
        (me4) edge (n4)
        ;
    \end{tikzpicture}
    $\quad\Rightarrow\quad$
    \begin{tikzpicture}[baseline={([yshift=-.5ex]current bounding box.center)},
      shorten >=1pt,node distance=12mm and 8mm,inner sep=1pt,on grid,auto]
      \node[]   (s1)                 {$S_1$};
      \node[right=of s1]   (s2)                 {$S_2$};
      \node[right=of s2]   (s3)                 {$S_3$};
      \node[right=of s3]   (s4)                 {$S_4$};
      \node[right=of s4]   (s5)                 {$S_5$};

      \node[below=of s1,xshift=4mm]   (n1)        {$S'_1$};
      \node[right=of n1]   (n2)        {$S'_2$};
      \node[right=of n2]   (n3)        {$S'_3$};
      \node[right=of n3]   (n4)        {$S'_4$};
      \path[->]
        (s1) edge (n1)
        (s2) edge node[swap,near start,yshift=-1mm] {$\scriptstyle M$}(n1)
        (s2) edge (n2)
        (s3) edge node[swap,near start,yshift=1.5mm,xshift=-1pt] {$\scriptstyle M$} (n3)
        (s4) edge node[swap,near start,yshift=-1mm] {$\scriptstyle M$} (n3)
        (s5) edge node[swap,near start,yshift=-1mm] {} (n4)
      ;
    \end{tikzpicture}

  \end{center}
  \caption{On the left, an abstract illustration of some transition of the determinization
  construction is presented, where the nodes represent non-empty sets and edges show the
  relationship between sets. In $\op{step}$ the successors are separated, in $\op{prune}$
  the successor sets (wrt. the normalized successor relation) which are empty are removed, in
  $\op{merge}$ adjacent non-empty sets are merged and in $\op{normalize}$ the final ranks are
  assigned. On the right, the resulting run-DAG edges for this transition are depicted,
  marked edges are denoted with $M$ and indicate that only accepting successor states of
  states in the source set are contributed to the target set of the edge.}

  \label{fig:rundag}
\end{figure}

See Figure~\ref{fig:rundag} for an illustration of the relationship of the determinization
transitions and the run-DAG.
We also need the following related definitions, collected here for reference:

\begin{definition}
  \label{def:dpanbadefs}

  Let $w\in L(\mc{B})$ and let $k$ be the smallest rank that is infinitely often green and
  finitely often red during the run of $\mc{B}$ on $w$, i.e. the rank witnessing the acceptance.

  Let $i_0$ be some time after which no rank $<k$ is ever active again and $k$ is never red again.

  Let $r_i$ denote the tuple index of the set with rank $k$ at time $i$, and let $l_i <
  r_i$ be the index of the \emph{bad left neighbor}, defined as the maximal index less
  than $r_i$ such that there is a bad path of the run-DAG starting in the set $S_{i,l_i}$
  (if no such left bad neighbor exists, let $l_i := 0$).

  Finally, let $\op{good}(i) := \bigcup_{j=l_i+1}^{r_i} S_{i,j}$.

\end{definition}

Intuitively, if the DPA accepts a word, this means that the smallest rank $k$ which is
infinitely often active is only finitely often red, and therefore will not be reassigned
to a completely unrelated set of states in the ranked slices.
We show that $k$ must eventually mark
sets containing states of at least one accepting run infinitely often. Unfortunately,
considering only the sets with rank $k$ in isolation is not sufficient to verify that
this is the case, but the following Lemma~\ref{lem:goodpieces} gives
us a sequence of ``checkpoints'' and sets such that we can obtain an infinite sequence of
suitable pieces to construct an accepting NBA run using König's Lemma.
An illustration of this technical result is found in Figure~\ref{fig:intervallschlauch}.

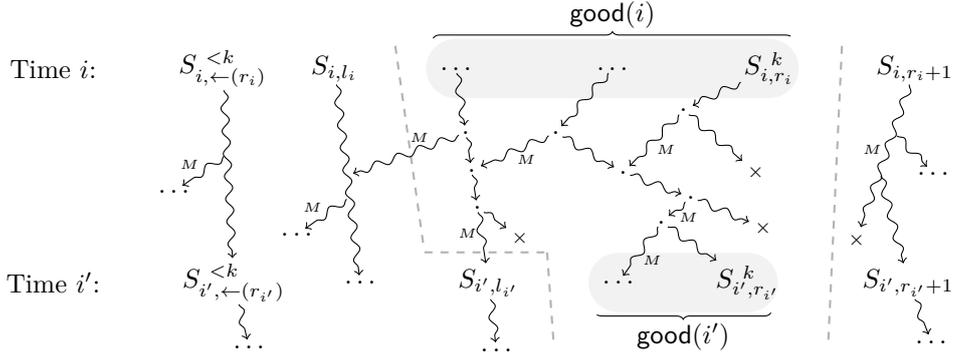
\begin{figure}[!htbp]
  \begin{center}
    \begin{tikzpicture}[baseline={([yshift=-.5ex]current bounding box.center)},
      shorten >=1pt,node distance=1cm,inner sep=1pt,auto]
      \node (labi) {Time $i$:};
      \node[below=2.5cm of labi] (labi2) {Time $i'$:};

      \node[right=of labi] (lsib1) {$S_{i,\lsibling(r_i)}^{\ <k}$};
      \node[right=5mm of lsib1] (bln1) {$S_{i,l_i}$};
      \node[right=of bln1] (anon1) {$\ldots$};
      \node[right=3.5cm of anon1] (kr1) {$S_{i,r_i}^{\ k}$};
      \node[right=of kr1] (rn1) {$S_{i,r_i+1}$};

      \node[right=of labi2] (lsib2) {$S_{i',\lsibling(r_{i'})}^{\ <k}$};
      \node[right=7mm of lsib2] (bln2) {$\ldots$};
      \node[right=of bln2] (anon2) {$S_{i',l_{i'}}$};
      \node[right=2.5cm of anon2] (kr2) {$S_{i',r_{i'}}^{\ k}$};
      \node[right=of kr2] (rn2) {$S_{i',r_{i'}+1}$};

      \node (anon3) at ($(anon1)!0.5!(kr1)$) {$\ldots$};
      \node (anon4) at ($(anon1)!0.3!(anon2)$) {.};
      \node (anon5) at ($(anon4)!0.5!(anon2)$) {.};
      \node[xshift=5mm] (anon6) at ($(anon5)!0.4!(anon2)$) {$\scriptstyle\times$};
      \node (anon7) at ($(anon2)!0.5!(kr2)$) {$\ldots$};
      \node (anon8) at ($(anon4)!0.5!(anon5)$) {.};

      \node (tmp1) at ($(lsib1)!0.4!(lsib2)$) {};
      \node[yshift=-5mm,xshift=-7mm] (tmpD1) at (tmp1) {$\ldots$};
      \node (tmp2) at ($(bln1)!0.6!(bln2)$) {};
      \node (tmp2b) at ($(bln1)!0.5!(bln2)$) {};
      \node[yshift=-5mm,xshift=-7mm] (tmpD2) at (tmp2) {$\ldots$};

      \node[xshift=-2.5mm] (tmp3) at ($(rn1)!0.3!(rn2)$) {};
      \node[xshift=-4mm,yshift=6mm] (tmp4) at ($(rn1)!0.7!(rn2)$) {};
      \node (tmp5) at ($(rn1)!0.8!(rn2)$) {};
      \node[xshift=7mm] (tmp6) at (tmp4) {$\ldots$};
      \node[xshift=-7mm] (tmp7) at (tmp5) {$\scriptstyle\times$};

      \node[right=of anon4] (x1) {.};
      \node[right=1.5cm of x1,yshift=3mm] (x2) {.};
      \node[below right=of x1,yshift=3mm] (x3) {.};
      \node[right=1.5cm of x3] (x4) {$\scriptstyle\times$};
      \node[below right=of x3,yshift=5mm] (x5) {.};
      \node[below right=of x5,yshift=5mm] (x6) {$\scriptstyle\times$};
      \node[below left=of x5,yshift=5mm,xshift=5mm] (x7) {.};

      \node[below right=of anon2,yshift=1mm,xshift=-13mm] (x8) {$\ldots$};
      \node[below=of lsib2,yshift=5mm,xshift=2mm] (x9) {$\ldots$};
      \node[below=of rn2,yshift=5mm,xshift=3mm] (x10) {$\ldots$};

      \begin{scope}[on background layer]
        \draw[fill=black!5,draw=white] \convexpath{anon1,kr1}{4mm};
        \draw[decorate,decoration={brace}]                  ([xshift=-3mm,yshift=4mm] $(anon1)$)
          -- node[above=1mm] {$\op{good}(i)$}               ([xshift=3mm,yshift=4mm] $(kr1)$);

        \draw[fill=black!5,draw=white] \convexpath{anon7,kr2}{4mm};
        \draw[decorate,decoration={brace}]                  ([xshift=3mm,yshift=-4mm] $(kr2)$)
          -- node[below=1mm] {$\op{good}(i')$}              ([xshift=-3mm,yshift=-4mm] $(anon7)$);
      \end{scope}
      \draw[black!30, thick, dashed]
        ([yshift=3mm] $(bln1)!0.5!(anon1)$) -- ([yshift=4mm] $(bln2)!0.5!(anon2)$);
      \draw[black!30, thick, dashed]
        ([yshift=4mm] $(anon2)!0.45!(anon7)$) -- ([yshift=-8mm] $(anon2)!0.5!(anon7)$);
      \draw[black!30, thick, dashed]
        ([yshift=3mm] $(kr1)!0.5!(rn1)$) -- ([yshift=-8mm] $(kr2)!0.5!(rn2)$);
      \draw[black!30, thick, dashed]
        ([yshift=4mm] $(bln2)!0.5!(anon2)$) -- ([yshift=4mm] $(anon2)!0.45!(anon7)$);

      \draw[->] (bln1) edge[decorate, decoration={snake, segment length=3mm, amplitude=0.5mm,post length=1mm}] node {} (bln2);
      \draw[->] (lsib1) edge[decorate, decoration={snake, segment length=3mm, amplitude=0.5mm,post length=1mm}] node {} (lsib2);

      \draw[->] (rn1) edge[swap, decorate, decoration={snake, segment length=3mm, amplitude=0.5mm,post length=1mm}] node {$\scriptscriptstyle M$} (tmp7);
      \draw[->] (tmp3) edge[decorate, decoration={snake, segment length=3mm, amplitude=0.5mm,post length=1mm}] node {} (tmp6);
      \draw[->] (tmp4) edge[decorate, decoration={snake, segment length=3mm, amplitude=0.5mm,post length=1mm}] (rn2);

      \draw[->] (anon1) edge[decorate, decoration={snake, segment length=3mm, amplitude=0.5mm,post length=1mm}] (anon4);
      \draw[->] (anon4) edge[decorate, decoration={snake, segment length=3mm, amplitude=0.5mm,post length=1mm}] (anon8);
      \draw[->] (anon8) edge[decorate, decoration={snake, segment length=3mm, amplitude=0.5mm,post length=1mm}] (anon5);
      \draw[->] (anon4) edge[swap,near start, decorate, decoration={snake, segment length=3mm, amplitude=0.5mm,post length=1mm}] node {$\scriptscriptstyle M$} (tmp2b);
      \draw[->] (anon5) edge[decorate, decoration={snake, segment length=3mm, amplitude=0.5mm,post length=1mm}] (anon6);
      \draw[->] (anon5) edge[swap,near start, decorate, decoration={snake, segment length=3mm, amplitude=0.5mm,post length=1mm}] node {$\scriptscriptstyle M$} (anon2);

      \draw[->] (x1) edge[decorate, decoration={snake, segment length=3mm, amplitude=0.5mm,post length=1mm}] node {$\scriptscriptstyle M$} (anon8);
      \draw[->] (x1) edge[decorate, decoration={snake, segment length=3mm, amplitude=0.5mm,post length=1mm}] node {} (x3);
      \draw[->] (x2) edge[decorate, decoration={snake, segment length=3mm, amplitude=0.5mm,post length=1mm}] node {$\scriptscriptstyle M$} (x3);
      \draw[->] (anon3) edge[decorate, decoration={snake, segment length=3mm, amplitude=0.5mm,post length=1mm}] node {} (x1);
      \draw[->] (kr1) edge[decorate, decoration={snake, segment length=3mm, amplitude=0.5mm,post length=1mm}] node {} (x2);
      \draw[->] (x2) edge[decorate, decoration={snake, segment length=3mm, amplitude=0.5mm,post length=1mm}] node {} (x4);
      \draw[->] (x3) edge[decorate, decoration={snake, segment length=3mm, amplitude=0.5mm,post length=1mm}] node {} (x5);
      \draw[->] (x5) edge[decorate, decoration={snake, segment length=3mm, amplitude=0.5mm,post length=1mm}] node {} (x6);
      \draw[->] (x5) edge[decorate, decoration={snake, segment length=3mm, amplitude=0.5mm,post length=1mm}] node {$\scriptscriptstyle M$} (x7);
      \draw[->] (x7) edge[decorate, decoration={snake, segment length=3mm, amplitude=0.5mm,post length=1mm}] node {$\scriptscriptstyle M$} (anon7);
      \draw[->] (x7) edge[decorate, decoration={snake, segment length=3mm, amplitude=0.5mm,post length=1mm}] node {} (kr2);

      \draw[->] (anon2) edge[decorate, decoration={snake, segment length=3mm, amplitude=0.5mm,post length=1mm}] node {} (x8);
      \draw[->] (lsib2) edge[decorate, decoration={snake, segment length=3mm, amplitude=0.5mm,post length=1mm}] node {} (x9);
      \draw[->] (rn2) edge[decorate, decoration={snake, segment length=3mm, amplitude=0.5mm,post length=1mm}] node {} (x10);

      \draw[->] (tmp1) edge[swap,decorate, decoration={snake, segment length=3mm, amplitude=0.5mm,post length=1mm}] node {$\scriptscriptstyle M$} (tmpD1);
      \draw[->] (tmp2) edge[swap,decorate, decoration={snake, segment length=3mm, amplitude=0.5mm,post length=1mm}] node {$\scriptscriptstyle M$} (tmpD2);

    \end{tikzpicture}
  \end{center}
  \caption{An abstract sketch of the run-DAG, the different entities in Definition~\ref{def:dpanbadefs},
  and the statement of Lemma~\ref{lem:goodpieces}. On the levels $i,i' \ge i_0$, the sets with the smallest active rank
  $k$ ($S_{i,r_i}$ and $S_{i',r_{i'}}$) are illustrated with their bad left neighbors
  ($S_{i,l_i}$ and $S_{i', l_{i'}}$) from which bad paths in the run-DAG start. Since
  $S_{i,r_i}$ has rank $k$, its left subtree boundary set (in the tree view) has rank $<k$.
  Since no rank $<k$ is active after $i_0$, there must be a bad path starting from the
  left subtree boundary set of $S_{i,r_i}$, and hence the bad left neighbor of $S_{i,r_i}$
  is between (including) $S_{i,\lsibling(r_i)}$ and (excluding) $S_{i,r_i}$. The unions of the highlighted intervals
  form the sets $\op{good}(i)$ and $\op{good}(i')$, respectively.
  The time $i'$ is chosen such that all paths starting in the interval that defines $\op{good}(i)$ have
  used at least one marked edge of the run-DAG. If the bad left neighbor does not exist, then the
  interval defining the set $\op{good}$ goes up to the left border. Paths are marked with $M$ when they
  contain a marked edge, terminated paths in the run-DAG are marked with $\times$.
  Notice that paths can leave the interval boundary, but cannot enter.
  }
    \label{fig:intervallschlauch}
\end{figure}

\begin{lemma}
  \label{lem:goodpieces}
  Let $w, i_0$ and $\op{good}(i)$ be defined as in Definition~\ref{def:dpanbadefs}.
  Then for every time $i \geq i_0$ it holds that $\op{good}(i)\neq\emptyset$, and there is
  some $i'>i$ such that for every $q \in \op{good}(i')$ there is some $p \in \op{good}(i)$
  such that there is a path from $p$ to $q$ that is labelled by the substring $w(i)\ldots
  w(i'-1)$ of $w$ and visits at least one accepting state.
\end{lemma}
\begin{proof}
  Fix a time $i \geq i_0$ and let $k$, $r_i$ and $l_i$ also be defined as in
  Definition~\ref{def:dpanbadefs}.

  Clearly, $\op{good}(i)$ is not empty, as it at least contains the set $S_{i,r_i}$.
  Notice that by choice of time $i_0$ and rank $k$, the set with rank
  $k$ will never be merged together with some set to the right, i.e. a set with index
  $>r_i$, as this would violate the constraints of $\op{merge}$ (because after
  $i_0$ no rank $<k$ can be active). For the same reason, it will also be never merged
  with its left subtree boundary set (if any), which by definition has a smaller rank and
  therefore must be preserved unchanged.

  Next, observe that the set with rank $k$ can also never be merged with a set that lies
  on a bad path (and also cannot be a bad path itself), because lying on a bad path means
  that during $\op{step}$ the set of non-accepting normalized successors (the right child
  set) would never become empty (leading to the unmarked edges in the run-DAG), hence it
  would always inherit the rank $k$ and forever prevent $k$ from becoming green (which
  requires the right successor set $\Delta_{t_i}(S_{i,r_i},w(i))\cap\overline{F}$ to be
  empty), contradicting the choice of $k$.

  As after time $i_0$ no rank $<k$ will be active again, it means that the left subtree
  boundary set of $k$ (if it exists) will be neither red nor green ever again. But this implies, that in
  $\op{step}$ it always has a non-empty set of non-accepting successor states, which
  inherits its rank and therefore this set, by definition, must lie on a bad path in the
  run-DAG. Hence, $l_i \geq \lsibling_{t_i}(r_i)$.

  This motivates the definition of $\op{good}(i) := \bigcup_{j=l_i+1}^{r_i} S_{i,j}$ as
  the union of all sets whose successors could (in principle) be merged with the
  successors of $S_{i,r_i}$ without violating the assumptions.

  As, by definition, no set contained in $\op{good}(i)$ can lie on a bad path, it means that
  eventually every path of the run-DAG starting in some set inside $\op{good}(i)$ must
  either terminate or eventually go through a marked edge. Pick some $i' > i$ such that
  this is the case.

  Observe that tracing the left and right boundary indices $l_i$ and
  $r_i$ over time gives us an interval $(l_i,r_i]$ of indices for each level of the run-DAG
  such that no edge in the run-DAG from outside this interval can ever go inside this
  interval (if an edge enters from the left, then it contains a bad path and if an edge
  enters from the right, then the set with rank $k$ is merged with sets to the right, both
  cases we excluded).

  For this reason, every state $q \in \op{good}(i')$ must have been reached from some
  state $p\in \op{good}(i)$. As every path of the run-DAG starting in a set contained in
  $\op{good}(i)$ must go through at least one marked edge to reach a set in
  $\op{good}(i')$, it means that every path from a state $p\in \op{good}(i)$ to some state
  $q\in\op{good}(i')$ agreeing with the run-DAG must contain at least one accepting state,
  because a run-DAG edge is only marked, if only accepting successors of the source set
  end up in the target set of the edge.
\end{proof}

Using this result, we can show how an accepting run of the NBA can be constructed from an accepting
run of the DPA:

\begin{lemma}
  \label{lem:dpa_to_nba}
  $w\in L(\mc{B}) \implies w \in L(\mc{A})$
\end{lemma}
\begin{proof}
  Let $w\in L(\mc{B})$, rank $k$, time $i_0$ and sets $\op{good}(i)$ as
  defined in Definition~\ref{def:dpanbadefs}.

      Then, by Lemma~\ref{lem:goodpieces}, there is a set $\op{good}(i_0)$ that contains the
  current set with rank $k$ and thus is non-empty, and furthermore has the property that
  there exists some $i_1 > i_0$ such that $\op{good}(i_1)$ also is non-empty and all
  states in $\op{good}(i_1)$ are reached from at least one state in $\op{good}(i_0)$ by
  some path that visits some accepting state. By iteratively using
  Lemma~\ref{lem:goodpieces}, we can construct an infinite sequence of times $i_0, i_1,
  i_2, \ldots$ such that every pair of times $i_j$ and $i_{j+1}$ satisfies these
  properties.

  We now can construct a finitely branching infinite DAG in which the edges are labelled
  by finite run segments of the NBA.
  Level 0 of the DAG has one node for each initial state of the NBA.
  For $j \ge 1$, level $j$ has one node for each state in $\op{good}(i_{j-1})$.
  From level $0$ to level $1$, there is an edge from $q_0$ to $q$ if there is a run of $\mc{A}$ from $q_0$ to $q$ on the input $w(0) \cdots w(i_0-1)$.
  For $j \ge 1$, there is an edge from $p \in \op{good}(i_{j-1})$ on level $j$ to $q \in
  \op{good}(i_{j})$ on level $j+1$, if there is a run from $p$ to $q$ on the input $w(i_{j-1}) \cdots w(i_j -1)$ that visits an accepting state. Label the edge by this run.
   By Lemma~\ref{lem:goodpieces}, for each $q$ on level $j+1$ there is some $p$ on level $j$ such that there is an edge from $p$ to $q$.

   By König's lemma, there is an infinite path through this DAG starting on level 0. The concatenations of edge labels of this path
   yields a run of $\mc{A}$ on $w$ by
  construction. Starting from level 1, all edge labels contain at least one accepting
  state, and therefore this run must be accepting.
\end{proof}

By Lemmas~\ref{lem:nba_to_dpa} and \ref{lem:dpa_to_nba}, we have shown that $L(\mc{A}) =
L(\mc{B})$, which completes the proof of Theorem~\ref{thm:correctness}.
 }

\end{document}